\documentclass[11pt]{article}
\makeatletter
\providecommand{\@AC}[2]{}
\makeatother
\usepackage[utf8]{inputenc}
\usepackage{amsmath, amsfonts, amssymb, amsthm,  graphicx, mathtools, enumerate, float}
\usepackage[round]{natbib}
\usepackage[CJKbookmarks=true,
            bookmarksnumbered=true,  
			bookmarksopen=true,
			colorlinks=true,
			citecolor=blue,
			linkcolor=blue,
			anchorcolor=red,
			urlcolor=blue]{hyperref}
\usepackage{hyperref}
\usepackage{xcolor}
\usepackage{array}
\usepackage{thmtools}
\usepackage{thm-restate}
\usepackage{mathrsfs}
\usepackage{todonotes}
\newtheorem{thm}{Theorem}[section]
\newtheorem{remark}{Remark}[section]
\newtheorem{lemma}{Lemma}[section]

\newtheorem{proposition}{Proposition}[section]

\usepackage{caption}
\usepackage{subcaption}

\usepackage[blocks, affil-it]{authblk}
\usepackage{mathrsfs}
\usepackage{todonotes}
\usepackage[letterpaper, left=1.2truein, right=1.2truein, top = 1.2truein, bottom = 1.2truein]{geometry}
\usepackage[ruled, vlined, lined, commentsnumbered]{algorithm2e}
\renewcommand{\H}{\mathcal{H}}
\newcommand{\argmin}{\textnormal{argmin}}
\newcommand{\argmax}{\textnormal{argmax}}
\newcommand{\tol}{q}
\renewcommand{\th}{\textnormal{th}}
\newcommand{\lfdr}{\textnormal{lfdr}}
\renewcommand{\P}{\mathbb{P}}
\newcommand{\E}{\mathbb{E}}
\newcommand{\bFDR}{\textnormal{bFDR}}
\newcommand{\FDR}{\textnormal{FDR}}

\title{Adaptive procedures for boundary FDR control}
\author[1]{Sarah Mostow}
\author[2]{Daniel Xiang}
\affil[1]{Department of Statistical Science, Duke University}
\affil[2]{Department of Statistics, University of Chicago}
\date{\today}

\begin{document}

\maketitle

\begin{abstract}
    A cornerstone of the multiple testing literature is the Benjamini-Hochberg (BH) procedure, which guarantees control of the FDR when $p$-values are independent or positively dependent. While BH controls the average quality of rejections, it does not provide guarantees for individual discoveries, particularly those near the rejection threshold, 
    which are more likely to be false than the average rejection.
    For independent $p$-values with Uniform$(0,1)$ null distribution, the Support Line procedure (SL, \cite{soloff2024edge}) provably controls the error probability for the rejection at the edge of the discovery set (i.e. the one with largest $p$-value) at level $\tol m_0/m$, where $m_0$ is the number of true null hypotheses and $\tol$ is a tuning parameter. In this work, we study adaptive versions of the SL procedure that operate in two steps: the first step estimates $m_0$ from non-significant statistics, and the second step runs the SL procedure at an adjusted level $\tol m / \hat{m}_0$. The adaptive procedures are shown to control the false discovery probability for the ``boundary'' rejection under an independence assumption. 
    Simulation studies suggest that some but not all of the two-stage procedures maintain error control under positive dependence,
    and that substantial power is gained relative to the original SL procedure. 
    We illustrate differences between the procedures on meta-data from the recent literature in behavioral psychology on growth mindset and nudge interventions. 
\end{abstract}

\section{Introduction}

Statisticians have long deliberated how to determine, given a list of $p$-values arising from a simultaneous multiple testing experiment, which discoveries to further pursue.
Testing each null hypothesis at the traditional per-comparison error rate (PCER) $\alpha=0.05$ ignores the multiple testing problem and has led to replicability issues in psychological research
(\cite{open2015estimating,benjamin2018redefine}). Controlling the family-wise error rate (FWER) resolves the type 1 error inflation issue by using a more stringent rejection threshold,
but can be overly conservative in scientific applications
(see e.g. \cite{reiner2003identifying,risch1996future}). 
Middle ground is attained by the False Discovery Rate (FDR, \cite{benjamini1995controlling}), which measures the expected proportion of the discovery set that are false positives. By tolerating a fixed proportion of false rejections, the FDR approach allows the analyst to reject more null hypotheses than they would with a family-wise correction, while maintaining a meaningful form of type 1 error control. 

While the FDR approach corrects for multiplicity,
it also suffers from the “free-rider” problem: a large \textit{p}-value may be included in the rejection set not on its own merits but simply due to the presence of a few very small \textit{p}-values in the dataset. 
\cite{finner2001false} describe the ``problematic nature of the FDR concept", where by including some $p$-values that are known to exhibit strong evidence, the analyst may test their favored hypothesis at a more lenient significance level. 
\cite{soloff2024edge} proposed a correction using the Support Line (SL) procedure that controls the probability that its last rejection (the one with the largest $p$-value) is a false positive, below a user-specified level $\tol$. 
The right panel of Figure \ref{fig:bFDR-calibration} plots the false discovery probability of the last rejection, as a function of the tolerance parameter $0\leq \tol\leq 1$ for the BH (red) and SL (blue) procedures in a simulation setting from \cite{benjamini2006adaptive}. For $\tol=0.2$, the last discovery of the FDR procedure is over 50\% likely to be false, 
while the last discovery of the SL procedure at level $\tol=0.4$ is false only 20\% of the time. 
This error criterion, the probability that the last discovery is false, is called the boundary false discovery rate (bFDR, \cite{soloff2024edge,xiang2025frequentist}), and is reviewed along with the SL procedure in Section~\ref{sec:bFDR-background}.

Under independence assumptions, the procedures (BH, SL) control their respective error rates (FDR, bFDR) at level $\pi_0 \tol$, where $\pi_0$ denotes the proportion of null hypotheses that are true and $\tol$ is a tuning parameter encoding the analyst's tolerance for type 1 errors. Given $\pi_0$, these procedures could be run at level $\tol/\pi_0$, achieving the type 1 error level $\tol$ more closely while improving power. A substantial body of work has sought to incorporate an estimate of $\pi_0$ into FDR-controlling procedures by developing adaptive methods (see e.g. \cite{benjamini2000adaptive,storey2002direct,benjamini2006adaptive,gao2025adaptive}), which we review in Section \ref{sec:adaptive-FDR}. 
The general strategy is to construct a data-driven estimate $\hat{\pi}_0>0$ and then apply BH at level $\tol/\hat{\pi}_0$. Roughly speaking, if $1/\hat{\pi}_0$ is smaller than $1/\pi_0$ on average and $\pi_0$ is not too close to 1, this strategy increases power while maintaining $\FDR \leq \tol$. 

\begin{figure}[t]
    \centering
    \includegraphics[width=\linewidth]{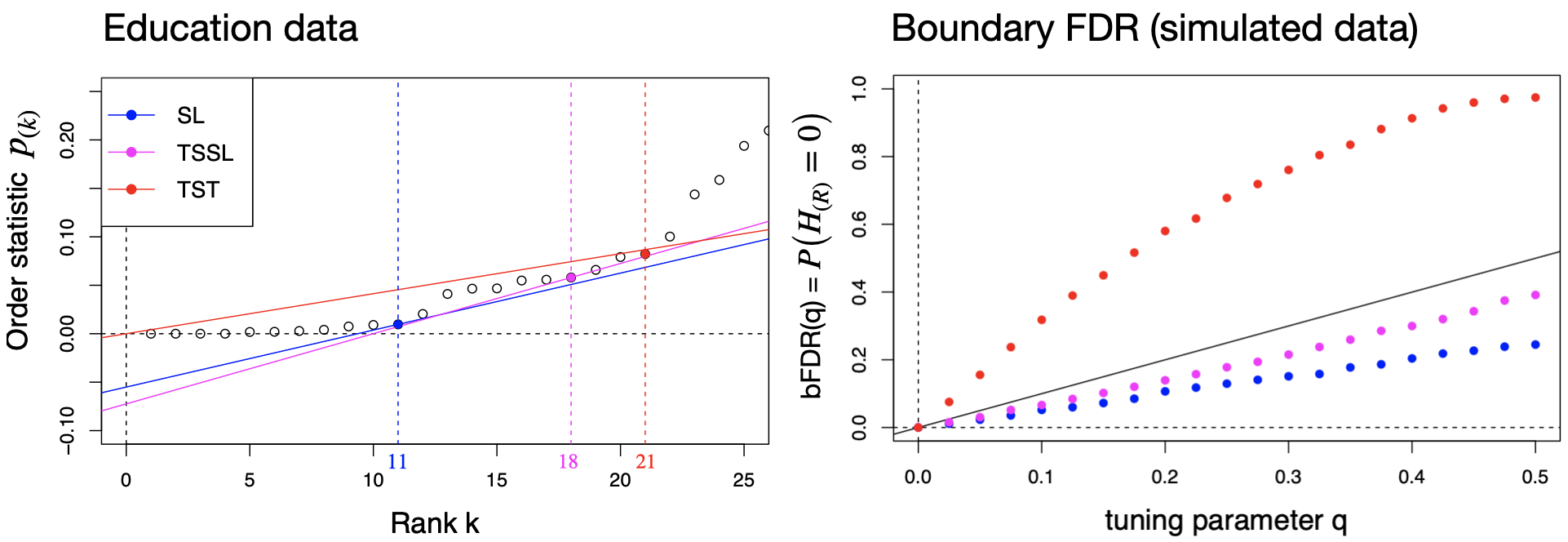}
    \caption{In the left panel, the TST$(\tol')$, SL$(\ell)$ and TSSL$(\ell')$ procedures are illustrated on an educational dataset from \cite{williams1999controlling}, with $\tol=0.1$ and $\ell=0.2$. 
    In the right panel, we plot the bFDR curve for the procedures in the ``1234 configuration'' simulation setting of \cite{benjamini2006adaptive}, where $\pi_0=1/2$ and non-null $p$-values are computed from a one-sided Gaussian location test with mean $\mu$ equal to 1, 2, 3, or 4 with equal weight. }
    \label{fig:bFDR-calibration}\
\end{figure}

In this paper, we study methods for bFDR control that similarly modify the SL procedure with an estimate of $\pi_0$. 
We provide theoretical results on their finite-sample bFDR control (Section \ref{sec:main}) and simulation studies demonstrating their performance gains relative to the standard SL procedure and its Storey-adjusted variant for independent and positively dependent test statistics (Section \ref{sec:simulation}). The Storey-type adjustment to SL with $\lambda=1/2$ works well under independence, but behaves erratically when there is positive dependence. Empirically, we find that using a smaller tuning parameter, e.g. $\lambda=0.2$, approximately preserves bFDR control for positively dependent statistics; a similar phenomenon was observed by \cite{blanchard2009adaptive} for the Storey-BH method.

Figure \ref{fig:bFDR-calibration} illustrates the Two-stage SL procedure (TSSL), an adaptive bFDR method that estimates the number of true null hypotheses by the number of non-significant $p$-values according to an initial run of the SL procedure. 
This procedure is illustrated along with the standard (unadjusted) SL procedure, and the Two-stage Step-up BH procedure (TST, \cite{benjamini2006adaptive}) on an educational dataset analyzed by \cite{benjamini2000adaptive} in the left panel of Figure \ref{fig:bFDR-calibration}. The TSSL procedure (with bFDR tolerance set to $\tol=0.2$) makes more rejections than SL, but still fewer than the adaptive BH procedure with FDR tolerance $\tol=0.1$. This reflects a general trend in our findings, that when the null proportion is not too large, the adaptive procedures lead to substantial power gains relative to the standard SL procedure.




\section{Background}
\label{sec:background}

\subsection{Boundary FDR and SL procedure}
\label{sec:bFDR-background}

Let $H_1,\dots,H_m$ denote the null hypotheses, $\mathcal{H}_0 \coloneqq \{i:H_i\text{ is true}\}$ the true null set, and suppose $p_1,\dots,p_m$ are random variables satisfying $p_i \sim \text{Uniform}(0,1)$ when $H_i$ is true. The Support Line procedure (SL, \cite{soloff2024edge}) is implemented as follows. First, compute
\begin{align}
\label{eq:sl}
    R = \argmax_{k=0,\dots,m} \left\{\frac{\tol k}{m}- p_{(k)} \right\}, 
\end{align}
and then let the index set of rejected null hypotheses be denoted $\mathcal{R}\coloneqq \{i: p_i \leq p_{(R)}\}$, where $p_{(0)}\coloneqq 0$. 
If $(p_i)_{i \in \H_0}$ are independent of each other and of $(p_i)_{i\notin \H_0}$, then
\begin{align}
\label{def:bfdr}
    \bFDR(\mathcal{R}) \coloneqq \P(H_{(R)}\text{ is true}, \; R > 0) = \pi_0 \tol,
\end{align}
where $\pi_0 \coloneqq |\H_0|/m$ and the notation $H_{(1)},\dots,H_{(m)}$ means the null hypotheses ordered according to $p_{(1)}\leq \dots \leq p_{(m)}$, where $H_{(0)}$ is false by definition. Note that while the truth statuses of $H_1,\dots,H_m$ may be fixed to begin with, once they are re-ordered according to $p_{(1)}\leq \dots \leq p_{(m)}$, the truth status of $H_{(i)}$ is a Bernoulli random variable whose `success' probability depends on the underlying $p$-value densities \citep{genovese2002operating}. 


A geometric interpretation of the SL procedure is illustrated in the left panel of Figure \ref{fig:bFDR-calibration}. On the plot of $p$-value order statistics, one can imagine pulling a line of slope $\tol/m$ up from height $-\infty$ until the plot of order statistics lays tangent to this line. The $x$-axis value at the point of tangency, where the plot is ``supported'' by this line, is the number of rejections made by SL.



\subsection{Adaptive procedures for FDR control}
\label{sec:adaptive-FDR}

\cite{benjamini2000adaptive} proposed the Lowest Slope (LSL) procedure to estimate $m_0$, the number of true null hypotheses. The method computes slopes $S_i$ of lines through $(m+1,1)$ and $(i,p_{(i)})$ for ordered $p$-values $p_{(1)}\leq p_{(2)}\leq \dots \leq p_{(m)}$, stopping at the first decrease in slope to yield $\hat{m}_0 = \min\{\lceil 1/S_{i} \rceil,m\}$ (a detailed description is provided in Section \ref{sec:main}). The resulting plug-in BH procedure is the original method designed for adaptive FDR control.

\cite{schweder1982plots} proposed an estimate of $\pi_0$  that is proportional to the number of $p$-values above a fixed threshold $\lambda$, a method later popularized by \cite{storey2002direct}. \cite{dohler2023unified} presented a general framework for $\pi_0$ estimation that extends to the discrete $p$-value setting, establishing rigorous plug-in FDR control for the Storey estimator and related methods. 
\cite{gao2025adaptive} further refined the Storey estimator by introducing data-driven stopping rules for choosing $\lambda$ adaptively while preserving FDR control. 

The LSL estimate of $\pi_0$ is equivalent to an adaptive Storey adjustment with a data-dependent tuning parameter $\lambda = p_{(i)}$, since
\begin{align*}
    \hat{\pi}_0^\lambda \coloneqq \frac{1+\# \{i:p_i > \lambda\}}{m(1-\lambda)} \text{ with }\lambda=p_{(i)}\Rightarrow \hat{\pi}_0^{\lambda} = \frac{1+m-i}{m(1-p_{(i)})}
\end{align*}
is proportional to the reciprocal of the slope $S_i$ at step $i$. The LSL null proportion estimator is then equivalent to $\hat{\pi}_0^{\text{LSL}} = \min\{\hat{\pi}_0^{p_{(j)}},1\}$, where $j$ is the first time $\hat{\pi}_0^{p_{(i)}}$ stops decreasing, starting from $i=1$. 
For independent $p$-values, this procedure's FDR guarantee can be approximately justified using the martingale proof technique used by \cite{gao2025adaptive}.

\cite{benjamini2006adaptive} proposed the \textit{two-stage linear step-up} (TST) procedure and studied it alongside other adaptive FDR procedures. The TST method estimates $m_0$ as the number of non-significant $p$-values from BH run at a reduced level 
in the first stage, and re-runs BH at an adjusted level 
in the second stage. The TST procedure was shown to control FDR under independence, and also numerically to be robust to positive correlations between test statistics. Our goal in the remaining sections is similar: to derive the bFDR guarantee for a \textit{two-stage support line} (TSSL) procedure, defined in the next section, and to compare the adaptive SL procedures defined by various estimators of $\pi_0$ in simulations and on real data.


\section{The adaptive bFDR procedures}
\label{sec:main}

\paragraph{Two-stage SL procedure.} The two-stage support line (TSSL) procedure with tuning parameter $\tol$ is defined as follows.
\begin{enumerate}
    \item Run SL at level $\tol$, i.e. compute
    \begin{align*}
        R_1 \coloneqq \argmax_{k=0,\dots,m}\left\{\frac{\tol k}{m} -p_{(k)} \right\}.
    \end{align*}
    If $R_1=0$, stop and reject nothing; if $R_1=m$ reject all $m$ hypotheses and stop; otherwise, go to Step 2.
    \item Re-run SL at level $\tol m/(m-R_1)$, i.e. compute
    \begin{align*}
        R_2 \coloneqq \argmax_{k=0,\dots,m} \left\{\frac{\tol k}{m-R_1}-p_{(k)} \right\},
    \end{align*}
    and reject $H_{(1)},\dots,H_{(R_{2})}$.
\end{enumerate}
The procedure is equivalent to estimating $\pi_0$ using the proportion of non-rejections at the first stage, and re-running SL with the adjusted tuning parameter $\tol/\hat{\pi}_0$ in the second stage. This two-stage SL procedure is the `local' analogue to the TST procedure in \cite{benjamini2006adaptive} and controls its boundary FDR below an inflated error level, as our next result shows. We remark that this result is analogous to Theorem 1 in \cite{benjamini2006adaptive}, in which FDR control is proven for the two-stage BH procedure. A proof of Theorem \ref{thm:bfdr-main} is recorded in Section \ref{sec:proofs}.

\begin{thm}
\label{thm:bfdr-main}
    Let $H_1,\dots,H_m$ denote $m$ null hypotheses, with independent $p$-values $p_1,\dots,p_m$. Suppose that $p_i \sim \text{Uniform}(0,1)$ if $H_i$ is true. Then
    \begin{align*}
        \bFDR(\mathcal{R}_2) \leq \frac{\tol}{1-\tol},
    \end{align*}
where $\mathcal{R}_2 \coloneqq \{i : p_i \leq p_{(R_2)}\}$ is the rejection set for the two stage procedure, and the $\bFDR$ of a rejection set is defined in \eqref{def:bfdr}.
\end{thm}

\begin{remark}
    To control bFDR below $\tol$, one can run the TSSL procedure at the reduced level $\tol'\coloneqq \tol/(1+\tol)$ instead of $\tol$ in both stages. We note that this is analogous to the TST procedure of \cite{benjamini2006adaptive}, where BH is run at the reduced level $\tol'$ in both stages in order to (provably) control FDR below $\tol$.
\end{remark}



\paragraph{Adaptive Storey adjustment.} The Adaptive Storey (AS) null proportion estimator \citep{gao2025adaptive} is defined:
\begin{align*}
    \hat{\pi}_0 \equiv \hat{\pi}_0(p_1,\dots,p_m) := \frac{1+\# \{i\leq m: p_i > \hat{\lambda}\}}{m(1-\hat{\lambda})},
\end{align*}
where $\hat{\lambda}(p_1,\dots,p_m) \geq \tol$ is any stopping time with respect to the filtration:
\begin{align}
\label{eq:filtration}
    \mathcal{F}_t = \sigma\Big(\sum_{i=1}^m 1\{p_i>s\}: \tol \leq s \leq t\Big), \hspace{2em} t\geq \tol.
\end{align}
Consider the SL procedure with tuning parameter $\tol / \hat{\pi}_0$ defined in the following way:
\begin{align*}
    R_{\tol/\hat{\pi}_0} = \argmin_{k:p_{(k)}\leq \tol} \Big\{ p_{(k)} - \frac{\tol k}{\hat{\pi}_0 m} \Big\}.
\end{align*}
Our next result shows that this adjustment preserves bFDR control for any valid stopping time $\hat{\lambda}$. The proof is adapted from arguments in \cite{soloff2024edge} and \cite{gao2025adaptive}, and is recorded in Section \ref{sec:proofs}.
\begin{proposition}
\label{thm:AS-bFDR-control}
    Let $H_1,\dots,H_m$ be null hypotheses, and suppose $p_1,\dots,p_m$ are independent $p$-values satisfying $p_i \sim \text{Uniform}(0,1)$ when $H_i$ is true. If $\hat{\lambda}$ is a valid stopping time, then
\begin{align*}
    \bFDR(\mathcal{R}_{\tol/\hat{\pi}_0}) \leq \tol.
\end{align*}
\end{proposition}
\cite{gao2025adaptive} recommended searching for $\lambda$ over a grid, $\tol, \tol + \delta, \tol + 2\delta, \tol + 3\delta, \dots$ for some small number $\delta>0$, e.g. $\delta = 1\%$ or $2\%$. More specifically, to choose $\hat{\lambda}$, we successively check the value of $\hat{\pi}_0^{\lambda}$ for each $\lambda$ in the grid, starting from $\tol$ and incrementing $\lambda$ upward along the grid until the first time that the associated estimate $\hat{\pi}_0^{\lambda}$ stops decreasing, taking $\hat{\lambda}$ to be this grid point at which the first increase was observed. 
This describes a valid stopping rule with respect to the filtration \eqref{eq:filtration}, so Proposition \ref{thm:AS-bFDR-control} implies that the resulting Storey-adjusted SL procedure, with data dependent $\lambda$, controls its bFDR. 

\paragraph{Storey with fixed threshold.} \cite{soloff2024edge} adjust the SL procedure using the Storey estimator with $\lambda=1/2$ and prove that this adjustment preserves bFDR control under an independence assumption in a Bayesian two-groups model. Our simulations show this procedure works very well in the independent case relative to all other procedures considered here, but worse in settings with positive dependence.
In the latter case, we recommend using the smaller level $\lambda=\tol$, e.g. $\lambda=0.2$. This $\pi_0$ estimator was noted to work well under dependence for controlling FDR with an adaptive BH procedure \citep{blanchard2009adaptive}. We find an analogous phenomenon holds for SL and the bFDR (Section \ref{sec:simulation}).

\paragraph{Lowest slope (LSL) estimator.} The lowest slope estimate of $m_0$
proceeds as follows. First the $p$-values are sorted from smallest to largest $p_{(1)}\leq \dots \leq p_{(m)}$. Starting from $i=1$, at each step, the slope $S_i$ of the line passing through $(m+1,1)$ and $(i,p_{(i)})$ is computed and compared to the slope $S_{i-1}$ passing through $(m+1,1)$ and $(i-1,p_{(i-1)})$, with $p_{(0)}\coloneqq 0$. If $S_{i} > S_{i-1}$, then we consider the next smallest $p$-value $p_{(i+1)}$ and repeat. The iteration stops at the first decrease in slope, i.e. when $S_{i}<S_{i-1}$, and the estimate is then $\hat{m}_0 = \min\{\lceil 1/S_{i} \rceil,m\}$  (see Figure 2 in \cite{benjamini2000adaptive} for a visualization on a small dataset).

Letting $\hat{\pi}_0^{\text{LSL}} \coloneqq \hat{m}_0/m$, the bFDR guarantee of the SL procedure adapted with this estimate of $\pi_0$ can be approximately justified using Proposition \ref{thm:AS-bFDR-control}. The reason is that for small values of $\tol$, the adaptive Storey procedure with grid $\{\tol ,p_{(i^*)},p_{(i^*+1)},\dots\}$ where $i^* = \min\{j:p_j \geq \tol\}$, often coincides exactly with the LSL estimate, as the value of $\lambda$ chosen by the latter is usually larger than $\tol$ in practice. 



\section{Numerical illustration}
\label{sec:simulation}


In this section, we study the type I error and power of the adaptive bFDR procedures on simulated data generated as in the numerical settings of \cite{benjamini1995controlling} and \cite{benjamini2006adaptive}. Specifically, one-sided \textit{p}-values are computed from unit variance Gaussian test statistics whose means are set according to one of the following configurations.

\paragraph{Alternating configuration.} The non-null means are set to
\begin{align*}
\mu_1=5\times 1/4, \; \mu_2=5\times 2/4,\; \mu_3=5\times 3/4,\;\mu_4=5 \times 4/4,    
\end{align*}
repeating cyclically for the first $m_1$ coordinates, where $m_1$ is a multiple of 4. The null means are all set to zero. 

\paragraph{``All at 5'' configuration.} In this case, the non-nulls are strongly separated from the nulls: the non-null means are equal to 5, and the null means are all equal to 0, i.e.
\begin{align*}
    \mu_1=\dots=\mu_{m_1}=5, \; \mu_{m_1+1}=\dots=\mu_m=0.
\end{align*}

\paragraph{Overview.} 
In Section \ref{sec:type1-sim}, we assess the boundary FDR as a function of the tuning parameter $\tol$ for each procedure in the setting with independent $p$-values in both the alternating mean configuration and the all-at-5 configuration, and compare the power of these procedures in the alternating mean setting.
In Section \ref{sec:type1-dependence}, we compare the bFDR of the procedures and also the variability of the null proportion estimators when there is positive equicorrelation between test statistics.

All procedures listed below can be understood as thresholding an empirical Bayes estimate of the local false discovery rate (lfdr, \cite{efron2001empirical}), defined in a Bayesian model where $H_1,\dots,H_m$ are random, and $\lfdr(p_i) = \P(H_i = 0 \mid p_i)$. 
In Section \ref{sec:variability-lfdr}, we assess the variability of lfdr estimates implied by the adaptive procedures. In the current frequentist context, $\lfdr(t)$ can be roughly interpreted as the false discovery proportion among $p$-values near a given point \citep{xiang2025frequentist}. 
We also assess the variability of the true lfdr at the rejection thresholds. 
Unless otherwise specified, the number of tests is set to $m=64$, the null proportion is $\pi_0=0.75$, and the tuning parameter is $\tol=0.2$. The procedures we compare are as follows:

    \begin{enumerate}
        \item \textbf{TSSL($q$):} Two-Stage Support Line using $\tol$ at both stages, which has guaranteed bFDR $\leq \frac{\tol}{1-\tol}$ in the independent case.
        \item \textbf{TSSL($q'$):} Two-Stage Support Line using $q' := \frac{q}{1+q}$ at both stages, which guarantees bFDR $\leq \tol$ in the independent case.
        \item \textbf{Storey(1/2):} Storey-adjusted Support Line procedure using fixed $\lambda = 0.5$.
        \item \textbf{Storey($q$)}: Storey-adjusted Support Line procedure using fixed $\lambda = \tol$.
        \item \textbf{AS(0.1; $q$):} Adaptive-Storey Support Line procedure with threshold $\hat \lambda$ chosen as described in Section \ref{sec:main}, with $\delta = 0.1$.
        \item \textbf{AS(0.01; $q$):} Adaptive-Storey Support Line procedure using $\delta = 0.01$.
        \item \textbf{AS(0.1; 0.5):} Adaptive-Storey Support Line procedure using $\delta = 0.1$ but using a grid that starts from 0.5 instead of $\tol$.
        \item \textbf{LSL:} Support Line procedure run at level $\tol / \hat{\pi}_0$ with the Lowest Slope Estimator for~$\hat{\pi}_0$.
        \item \textbf{SL:} Support Line procedure at level $\tol$ (Benchmark).
        \item \textbf{Oracle:} Support Line at level $\tol / \pi_0$ (Benchmark).
    \end{enumerate}
To estimate the boundary FDR for each procedure, we run $N=10,000$ independent multiple testing experiments and record for each one the null status of the last rejection. The $\bFDR$ estimate for each procedure is the proportion of experiments where its last rejection was a false discovery.


\subsection{Independent \textit{p}-values}
\label{sec:type1-sim}

\begin{figure}[h!]
    \centering
    \begin{subfigure}[b]{0.40\textwidth}
        \centering
        \includegraphics[width=\linewidth]{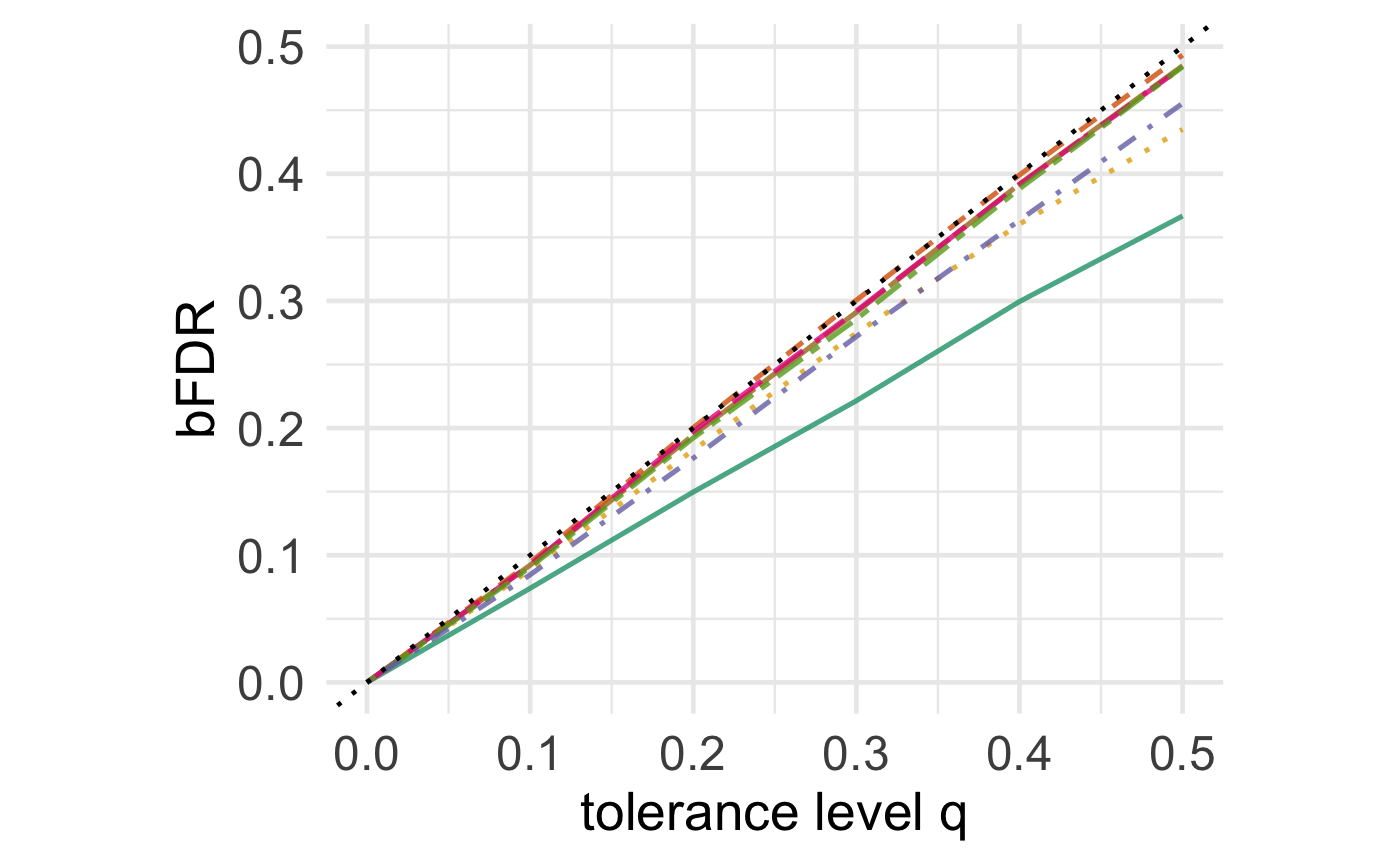}
        \caption{$\pi_0 = 0.75$}
        \label{fig:ind075 config 1}
    \end{subfigure}
    \begin{subfigure}[b]{0.40\textwidth}
        \centering
        \includegraphics[width=\linewidth]{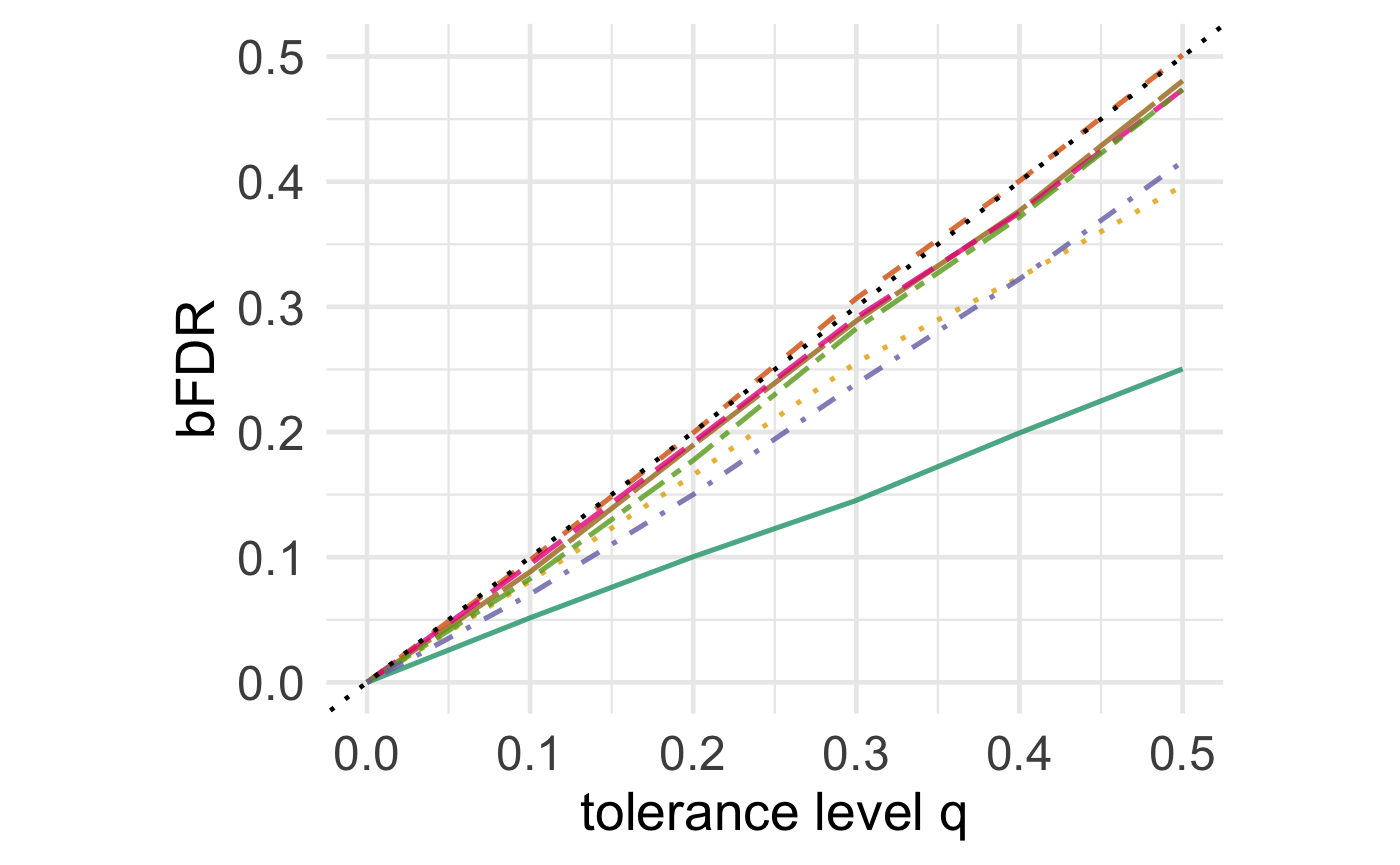}
        \caption{$\pi_0 = 0.5$}
        \label{fig:ind05 config 1}
    \end{subfigure}
    \begin{subfigure}[b]{0.18\textwidth}
        \centering
        \includegraphics[width=\linewidth]{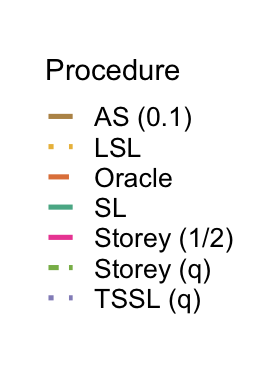}
        \vspace{0.2cm}
        \label{fig:legend}
    \end{subfigure}

    \caption{Boundary FDR versus tuning parameter for the alternating configuration (independent Gaussian test statistics), $N = 10{,}000$ simulations, $m = 64$ $p$-values, with $\pi_0 = 0.75$ (left) and $\pi_0 = 0.5$ (right). }
    \label{fig:ind config 1}
\end{figure}

\begin{figure}[h!]
    \centering
    \begin{subfigure}[b]{0.40\textwidth}
        \centering
        \includegraphics[width=\linewidth]{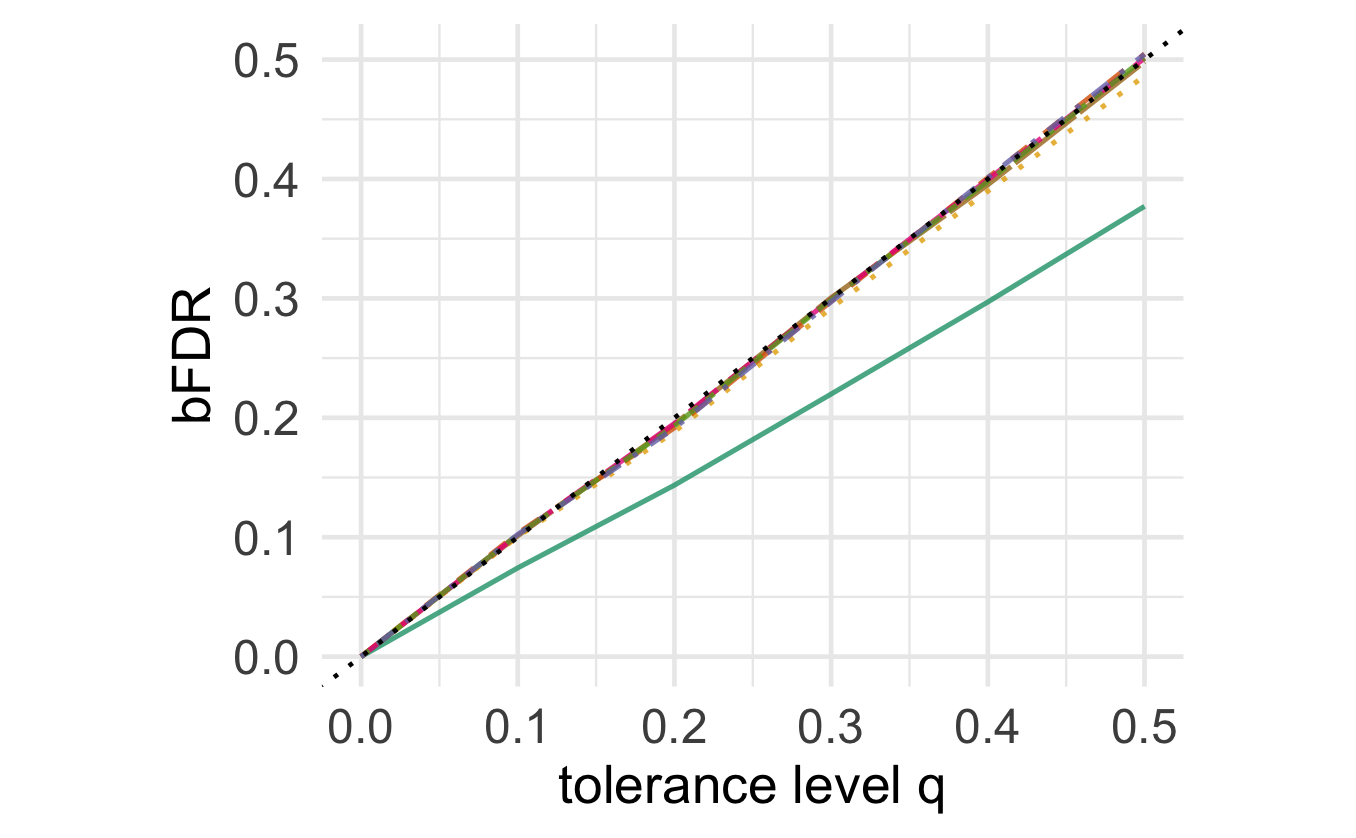}
        \caption{$\pi_0 = 0.75$}
        \label{fig:ind075 config 2}
    \end{subfigure}
    \begin{subfigure}[b]{0.40\textwidth}
        \centering
        \includegraphics[width=\linewidth]{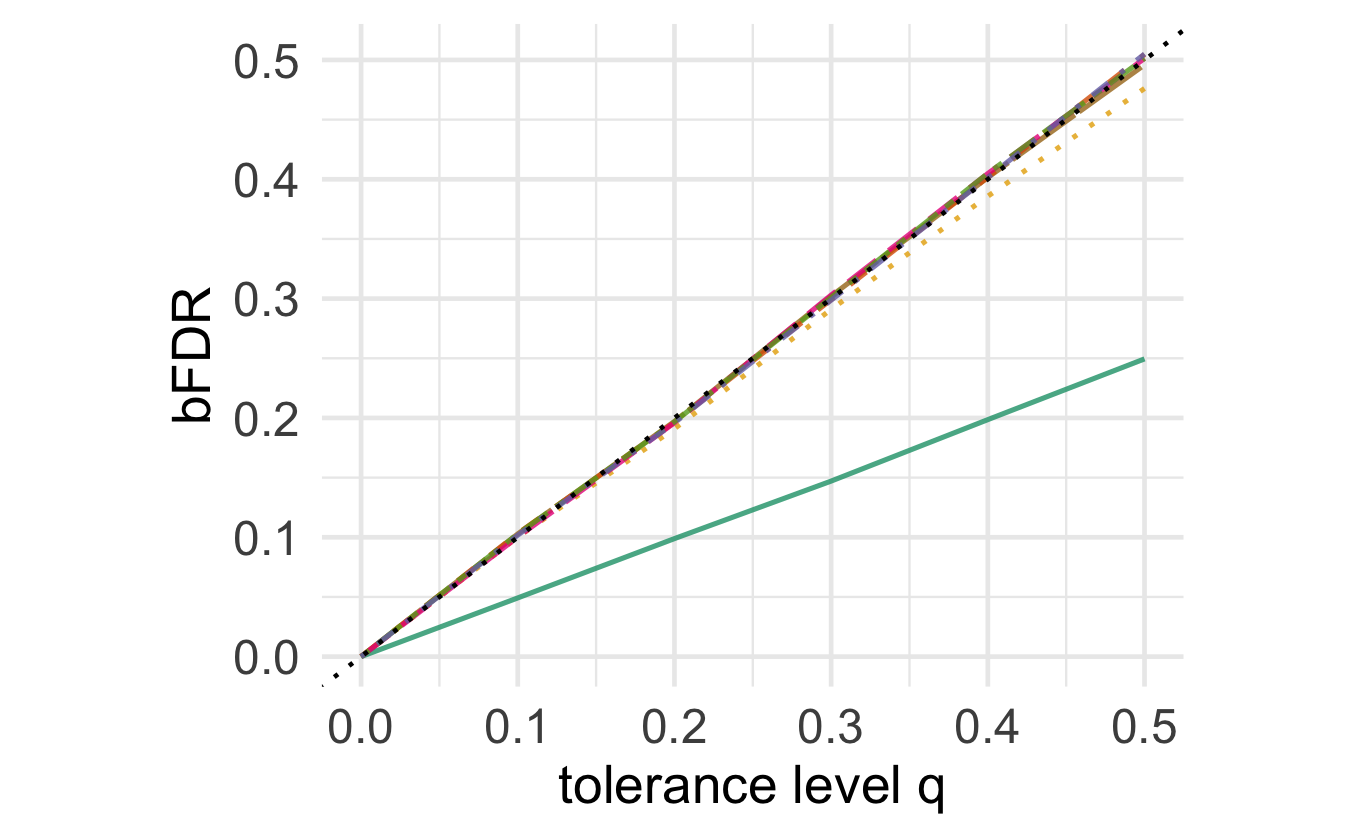}
        \caption{$\pi_0 = 0.5$}
        \label{fig:ind05 config 2}
    \end{subfigure}
    \begin{subfigure}[b]{0.18\textwidth}
        \centering
        \includegraphics[width=\linewidth]{figures/ind-legend.png}
        \vspace{0.2cm}
        \label{fig:legend 2}
    \end{subfigure}

    \caption{Boundary FDR versus tuning parameter for the all-at-5 configuration with $N = 10{,}000$ simulations, $m = 64$ p-values, with $\pi_0 = 0.75$ (left) and $\pi_0 = 0.5$ (right).}
    \label{fig:ind config 2}
\end{figure}
In Figures \ref{fig:ind config 1} and  \ref{fig:ind config 2}, we plot the $\bFDR$ versus the tolerance level $\tol$, for independent \textit{p}-values at two different levels of $\pi_0$, in both the 
alternating and all-at-5 configurations.
We can see that the $\bFDR$ for the adaptive procedures is much closer to the oracle $\bFDR$ and to $\tol$ compared to the standard Support Line procedure. The improvement is more pronounced when the alternative distribution is more extreme (as in the all-at-5 configuration) and when $\pi_0$ is smaller. Figure \ref{fig:bfdr vs pi0} shows that with larger values of $\pi_0$, the adaptive procedures still track the target bFDR $\tol=0.2$ much more closely than the standard Support Line procedure.

\begin{figure}[h!]
    \centering
    \begin{subfigure}[b]{0.40\textwidth}
        \centering
        \includegraphics[width=\linewidth]{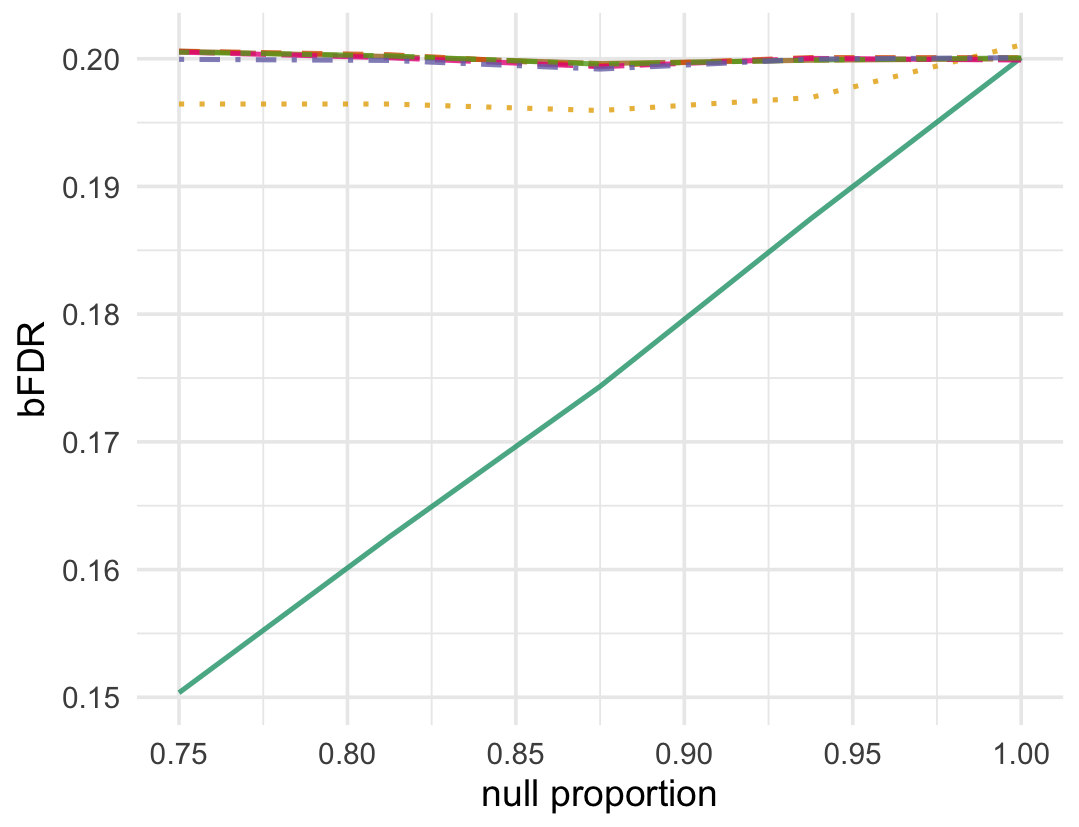}
        \caption{All-at-5 configuration}
        \label{fig:bfdr vs pi0 all 5}
    \end{subfigure}
    \begin{subfigure}[b]{0.40\textwidth}
        \centering
        \includegraphics[width=\linewidth]{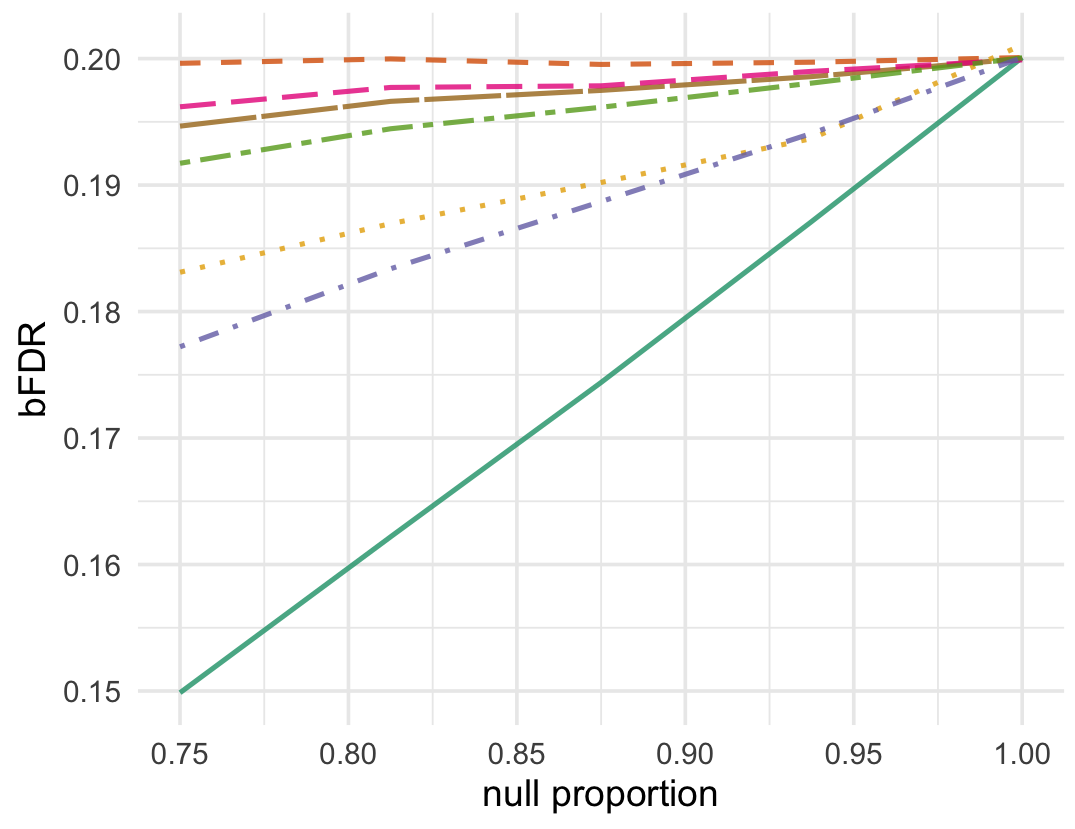}
        \caption{Alternating configuration}
        \label{fig:bfdr vs pi0 alt}
    \end{subfigure}
    \begin{subfigure}[b]{0.18\textwidth}
        \centering
        \includegraphics[width=\linewidth]{figures/ind-legend.png}
        \vspace{0.2cm}
        \label{fig:legend 3}
    \end{subfigure}

    \caption{Boundary FDR versus $\pi_0$ for the all-at-5 configuration (left) and alternating configuration (right) with $N = 10^6$ 
    simulations, $m = 64$ p-values, and $\tol = 0.2$. } 
    \label{fig:bfdr vs pi0}
\end{figure}

Figure \ref{fig:power_heatmaps}, which shows the power of each procedure relative to the oracle, also demonstrates that the biggest gains in power are achieved when $\pi_0$ is small. We also see that TSSL$(q')$ does not always outperform the Support Line procedure, particularly when $\tol$ is small and $\pi_0$ is close to 1. However, the rest of the adaptive procedures still perform as well or better in nearly all cases, even for large values of $\pi_0$.

\subsection{Correlated \textit{p}-values}
\label{sec:type1-dependence}

\begin{figure}[h!]
    \centering
    \begin{subfigure}[t]{0.42\textwidth}
        \centering
        \includegraphics[width=\linewidth]{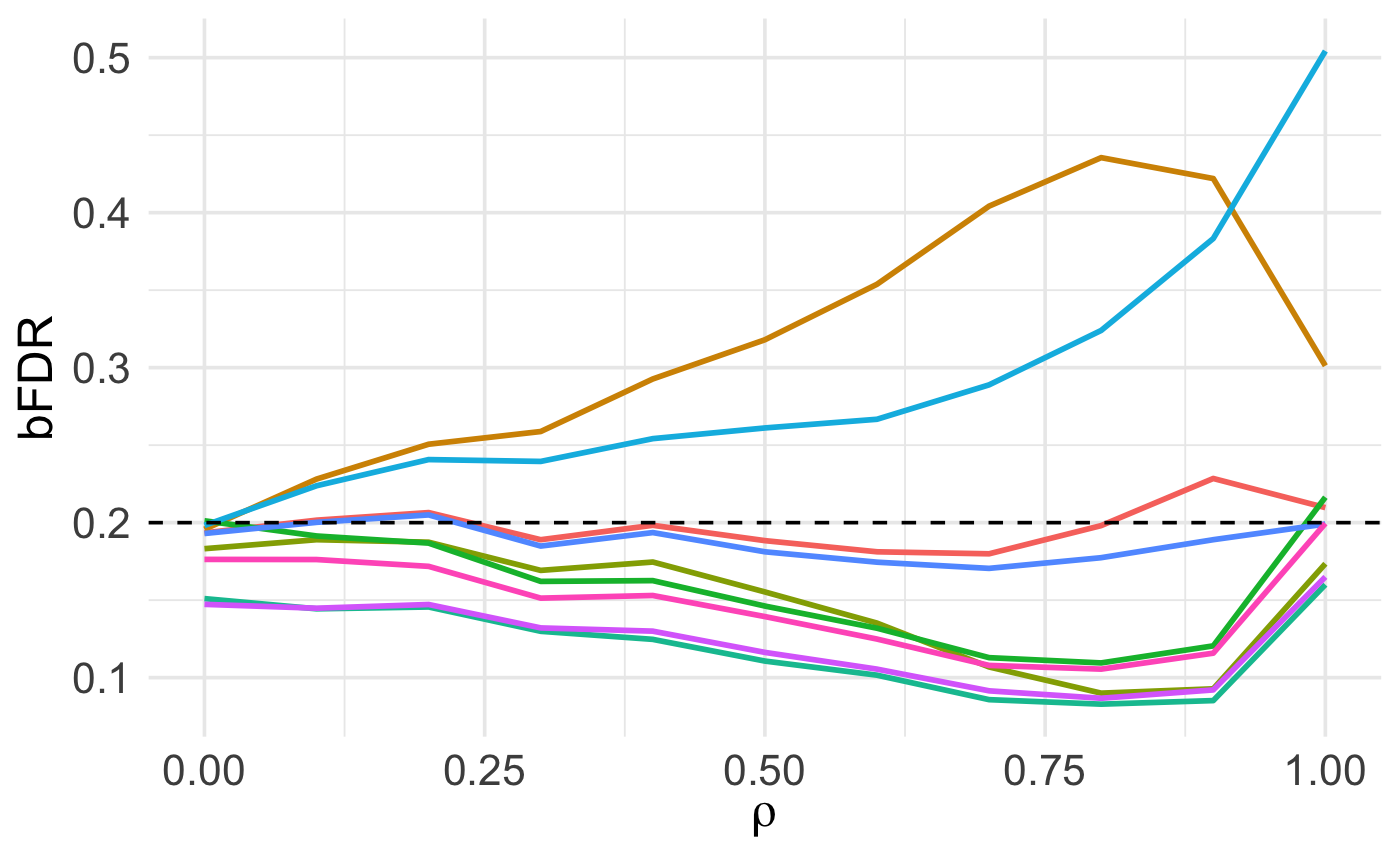}
        \caption{bFDR vs correlation $\rho$. }
        \label{fig:corr bfdr}
    \end{subfigure}
    \hfill
    \begin{subfigure}[t]{0.42\textwidth}
        \centering
        \includegraphics[width=\linewidth]{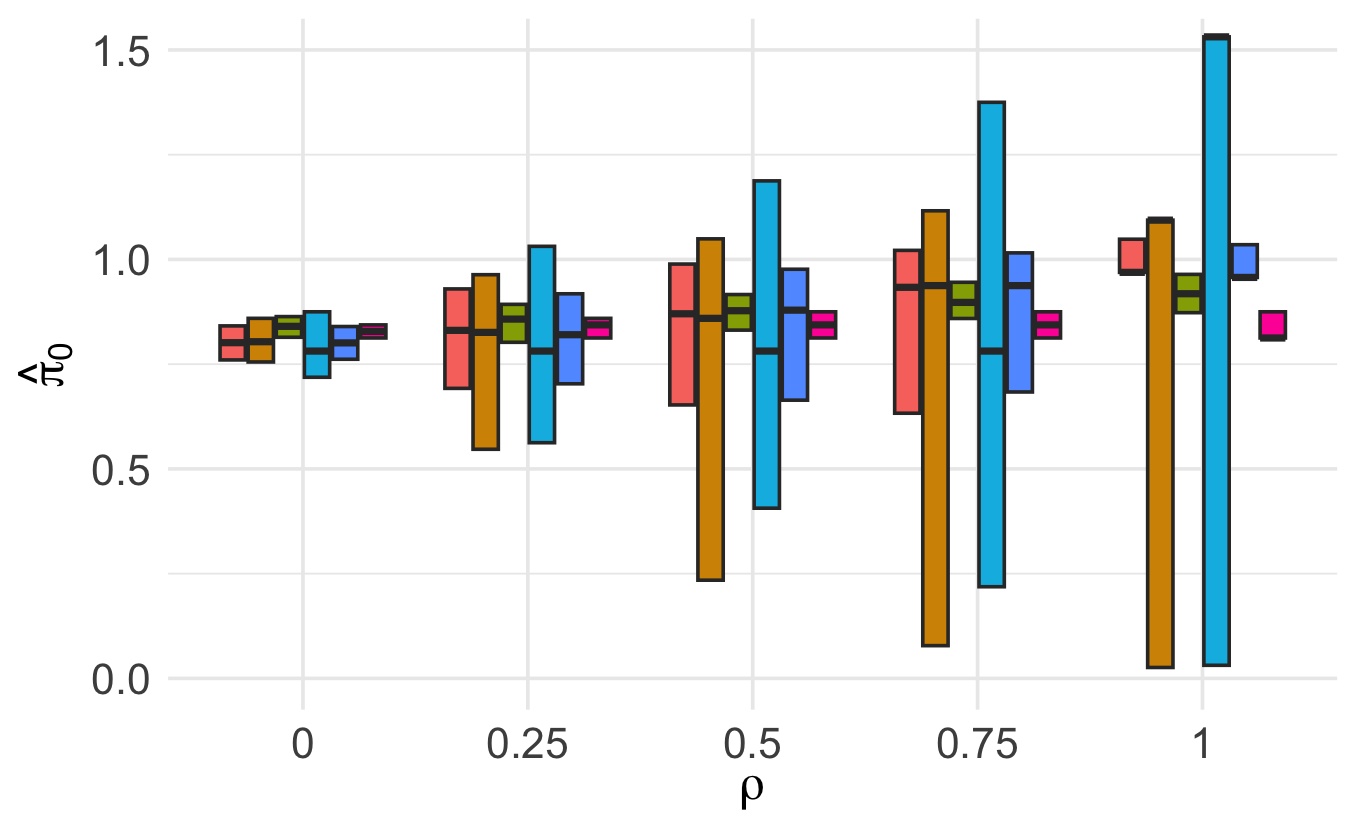}
        \caption{Interquartile ranges of $\hat{\pi}_0$ estimates for $\rho \in \{0, 0.25, 0.5, 0.75, 1\}$. }
        \label{fig:corr pi0hat}
    \end{subfigure}
    \hfill
    \raisebox{6mm}{
        \begin{subfigure}[t]{0.12\textwidth}
            \centering
            \includegraphics[width=\linewidth]{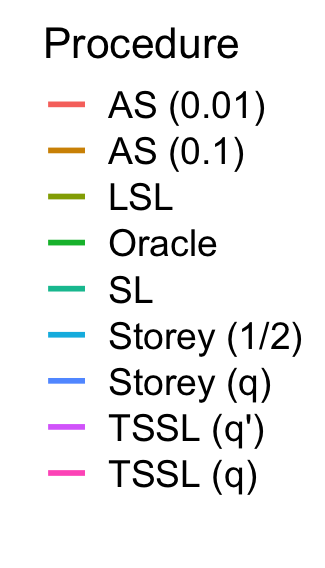}
            \label{fig:corr legend}
        \end{subfigure}
    }
    \caption{Boundary FDR (lines) and $\hat{\pi}_0$ (boxplots) versus correlation parameter for the alternating configuration (equicorrelated Gaussian test statistics), $N = 10{,}000$ simulations, $m = 64$ \textit{p}-values, $\tol = 0.2$, and $\pi_0 = 0.75$.}
    \label{fig:equicorr}
\end{figure}

Figure \ref{fig:corr bfdr} shows $\bFDR$ versus the dependence parameter $\rho$ for the same numerical setting as in Section \ref{sec:type1-sim} but with positively equicorrelated Gaussian noise, where the equicorrelation parameter $\rho$ varies between 0 and 1. Figure \ref{fig:corr pi0hat} suggests that the adaptive procedures' performance under strong positive dependence can be somewhat explained by the variability of their $\hat \pi_0$ estimates, since the procedures with the greatest $\bFDR$ violation also seem to have far higher variance in estimating $\pi_0$. Under strong positive dependence, the standard Storey-type adjustment (with $\lambda=1/2$) fails to maintain bFDR control. Futhermore, while the choice of $\delta$ for the Adaptive Storey procedure did not seem to make a considerable difference for $\bFDR$ control in the independent case, we see a bigger
difference in the correlated case.
On the other hand, the Storey adjustment with fixed $\lambda=\tol=0.2$ does not follow this trend, and appears to perform well across all values of $\rho$. This likely also explains the success of Adaptive Storey with $\delta=0.01$, which stops early and thus performs most similarly to Storey with fixed $\lambda=\tol$.
\begin{figure}[htbp]
    \centering

    \begin{subfigure}{0.95\textwidth}
        \centering
        \includegraphics[width=\textwidth]{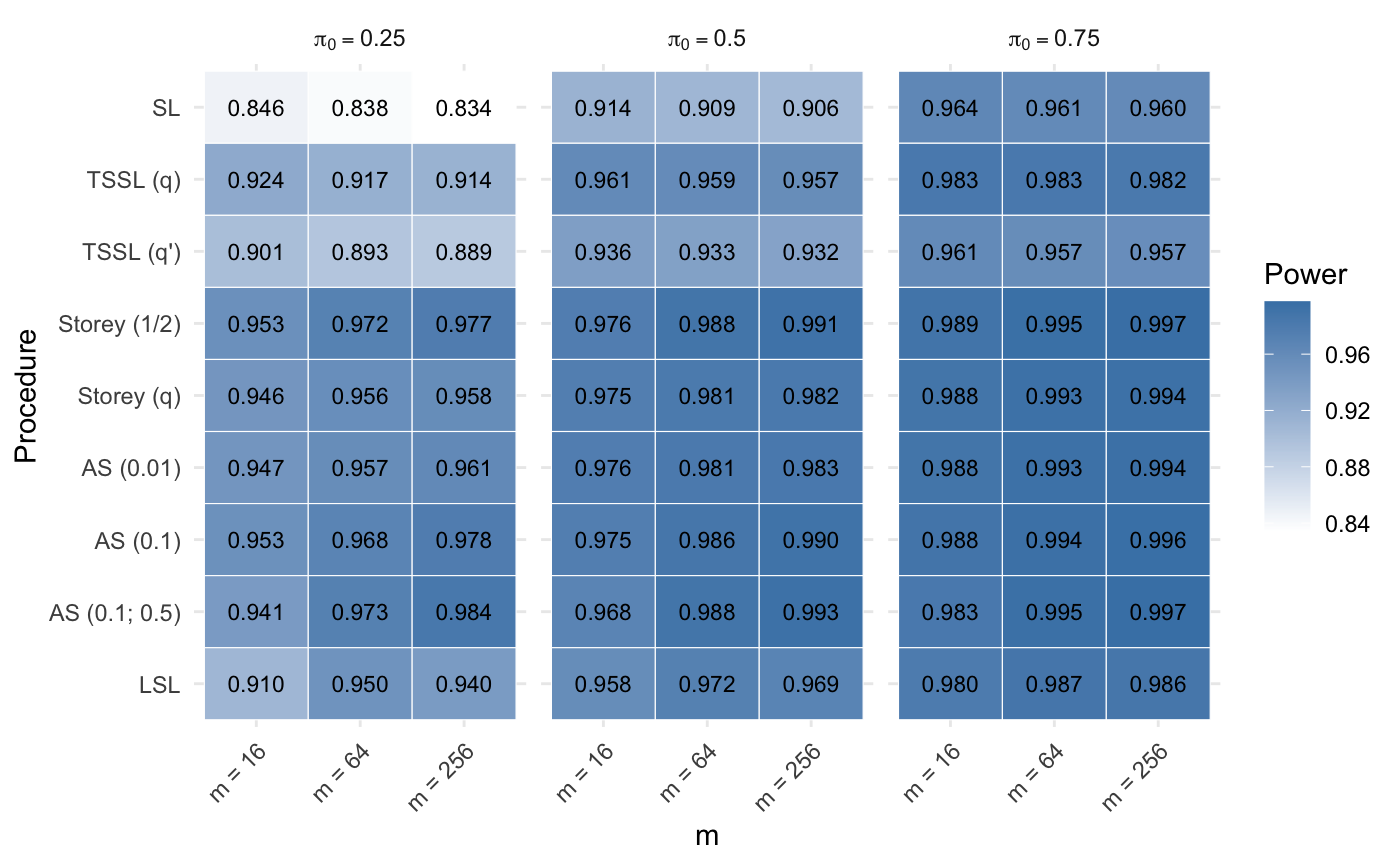}
        \caption{}
    \end{subfigure}

    \vspace{0.2cm}

    \begin{subfigure}{0.95\textwidth}
        \centering
        \includegraphics[width=\textwidth]{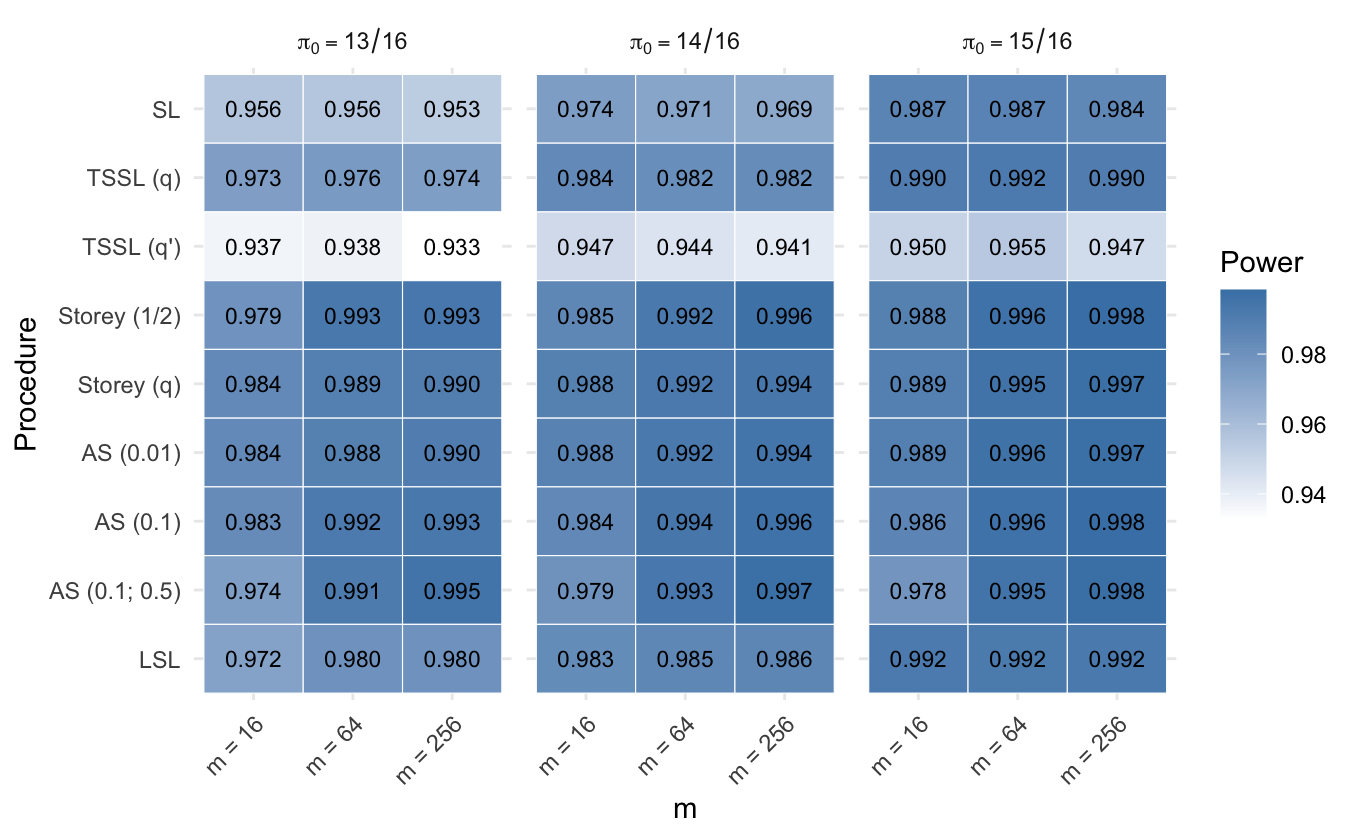}
        \caption{}
    \end{subfigure}

    \caption{Power heatmaps showing power of multiple testing procedures relative to the oracle procedure, using $\pi_0 = 0.25, 0.5, 0.75, 13/16, 14/15, 15/16$, $m = 16, 64, 256$, and $\tol = 0.2$. $N = 10,000$ simulations. \ }
    \label{fig:power_heatmaps}
\end{figure}

\subsection{Variability of lfdr}
\label{sec:variability-lfdr}
\begin{figure}[htbp]
\centering

\begin{subfigure}[b]{0.3\textwidth}
    \centering
    \includegraphics[width=\linewidth]{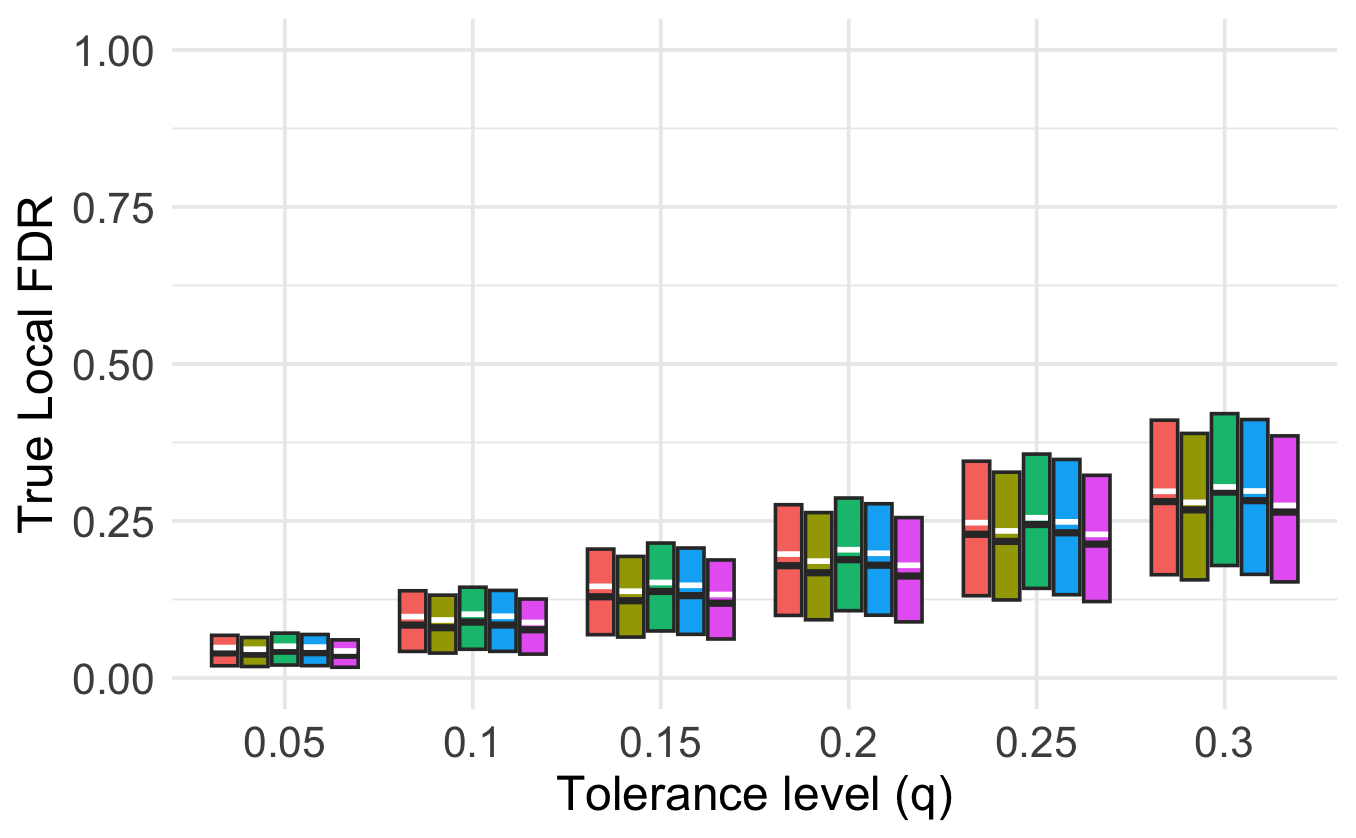}
    \caption{$\lfdr(p_{(R)})$ ($m=64$)}
\end{subfigure}%
\hspace{0.01\textwidth}
\begin{subfigure}[b]{0.3\textwidth}
    \centering
    \includegraphics[width=\linewidth]{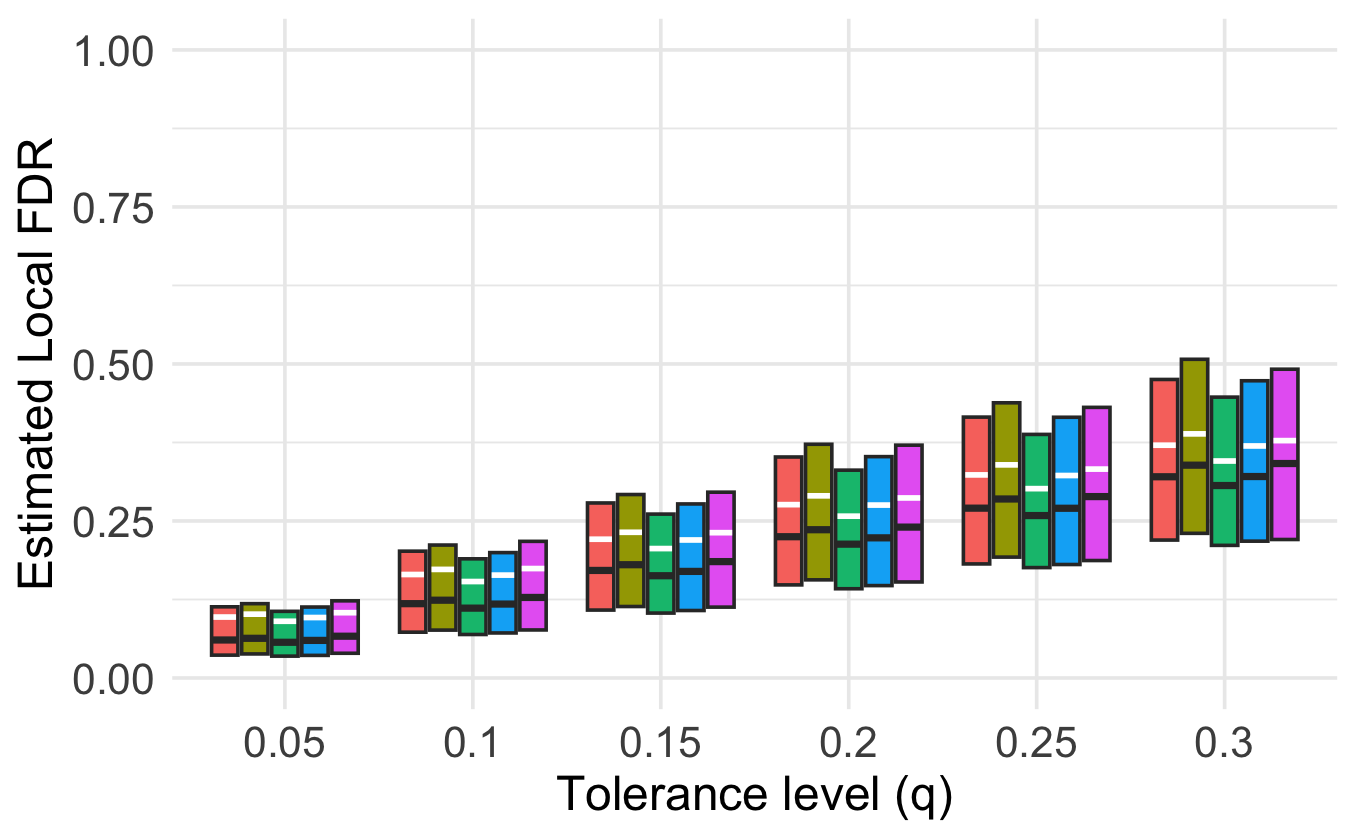}
    \caption{$\widehat{\text{lfdr}}(t^*)$ ($m=64$)}
\end{subfigure}%
\hspace{0.01\textwidth}
\begin{subfigure}[b]{0.3\textwidth}
    \centering
    \includegraphics[width=\linewidth]{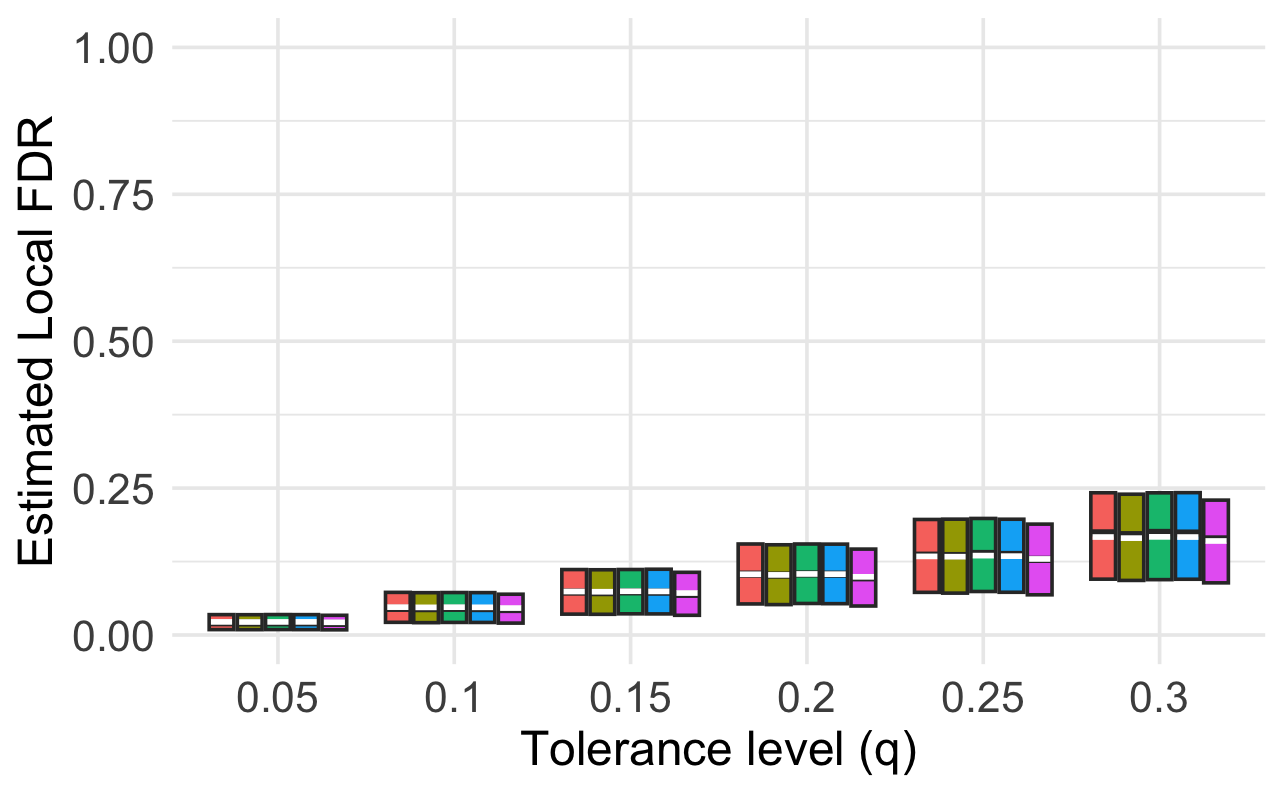}
    \caption{$\widehat{\text{lfdr}}(p_{(R)})$ ($m=64$)}
\end{subfigure}%
\hfill
\begin{minipage}[c]{0.06\textwidth}
    \centering
    \includegraphics[width=\linewidth]{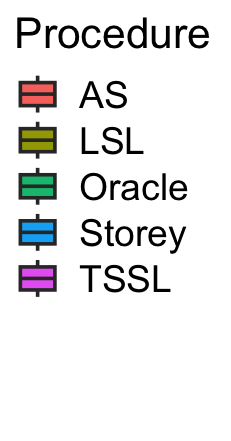}
\end{minipage}

\vspace{1em}

\begin{subfigure}[b]{0.3\textwidth}
    \centering
    \includegraphics[width=\linewidth]{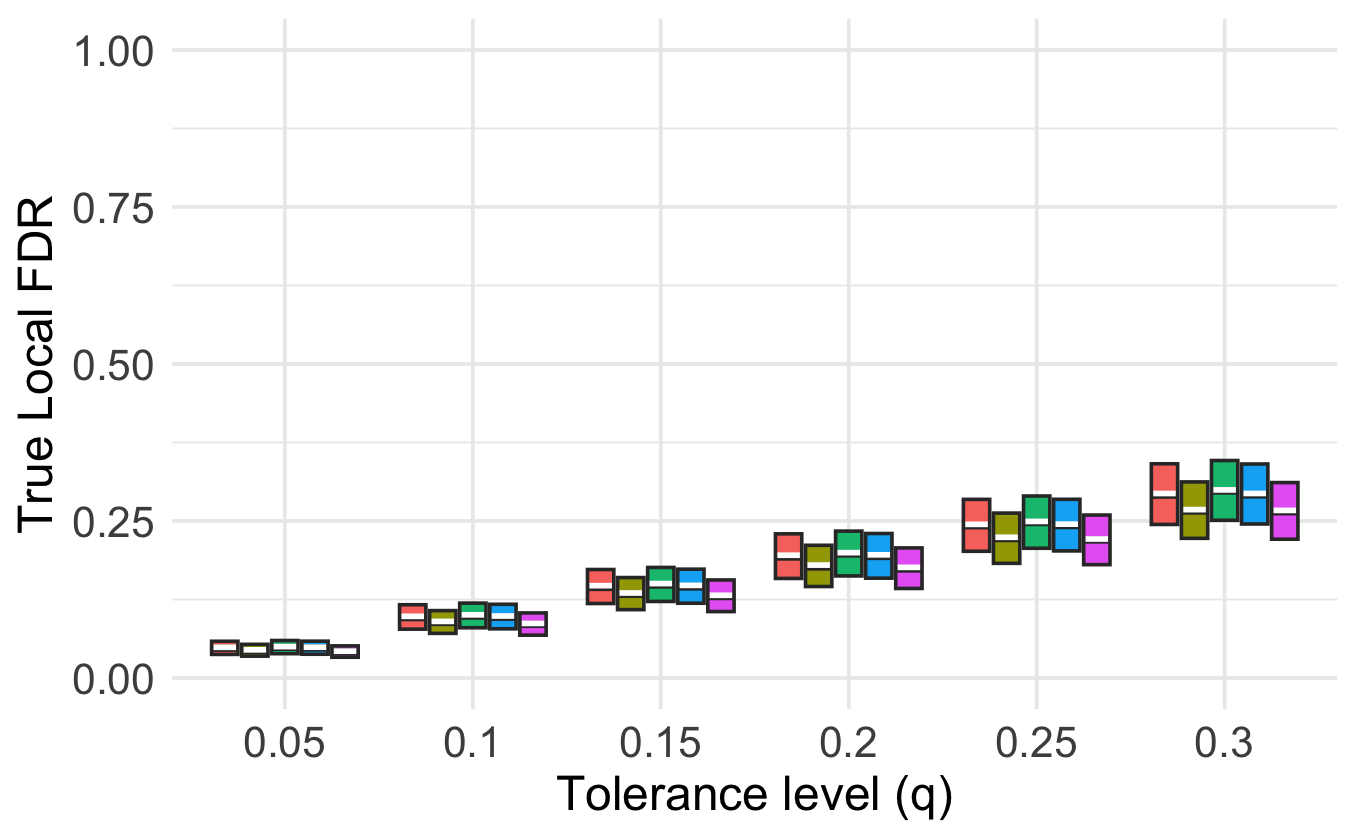}
    \caption{$\lfdr(p_{(R)})$ ($m=1024$)}
\end{subfigure}%
\hspace{0.01\textwidth}
\begin{subfigure}[b]{0.3\textwidth}
    \centering
    \includegraphics[width=\linewidth]{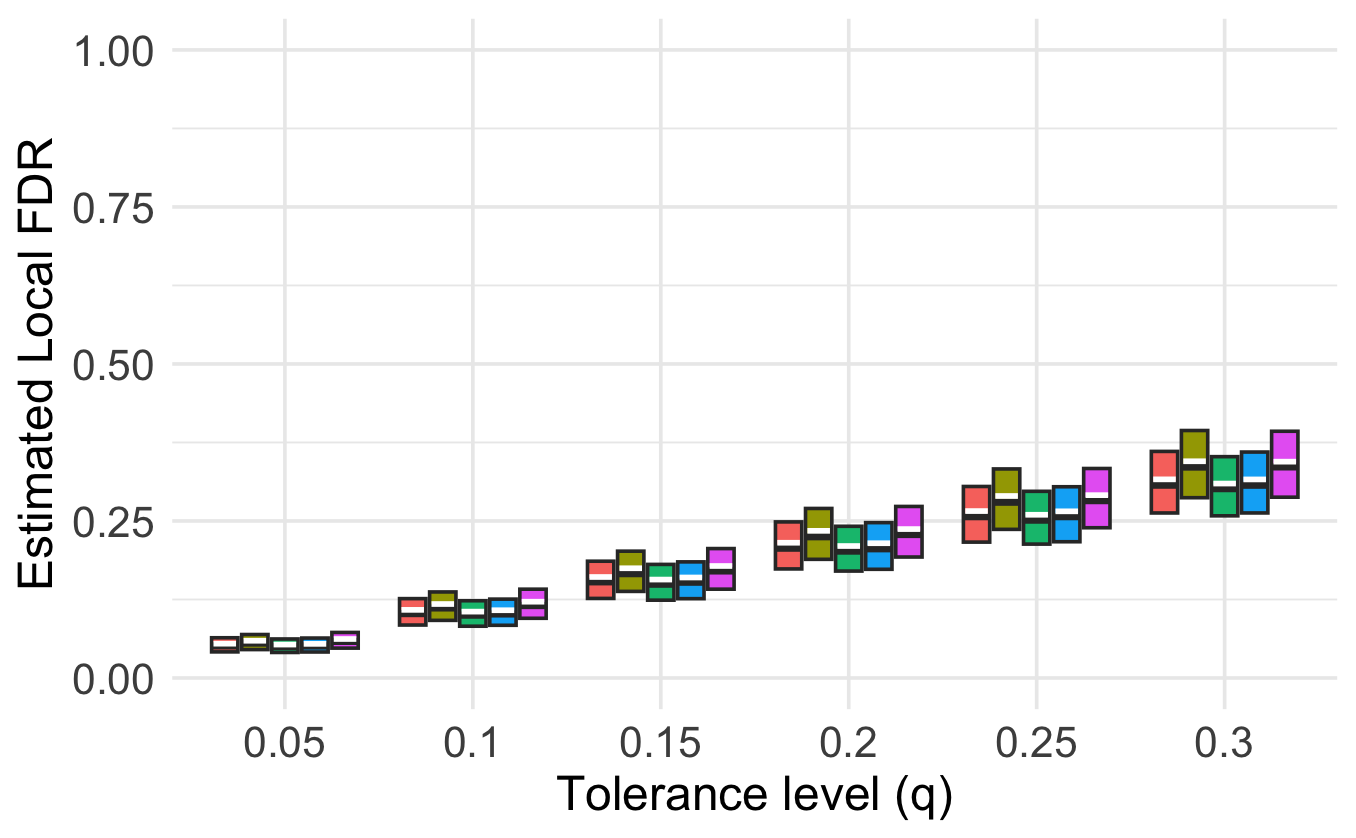}
    \caption{$\widehat{\text{lfdr}}(t^*)$ ($m=1024$)}
\end{subfigure}%
\hspace{0.01\textwidth}
\begin{subfigure}[b]{0.3\textwidth}
    \centering
    \includegraphics[width=\linewidth]{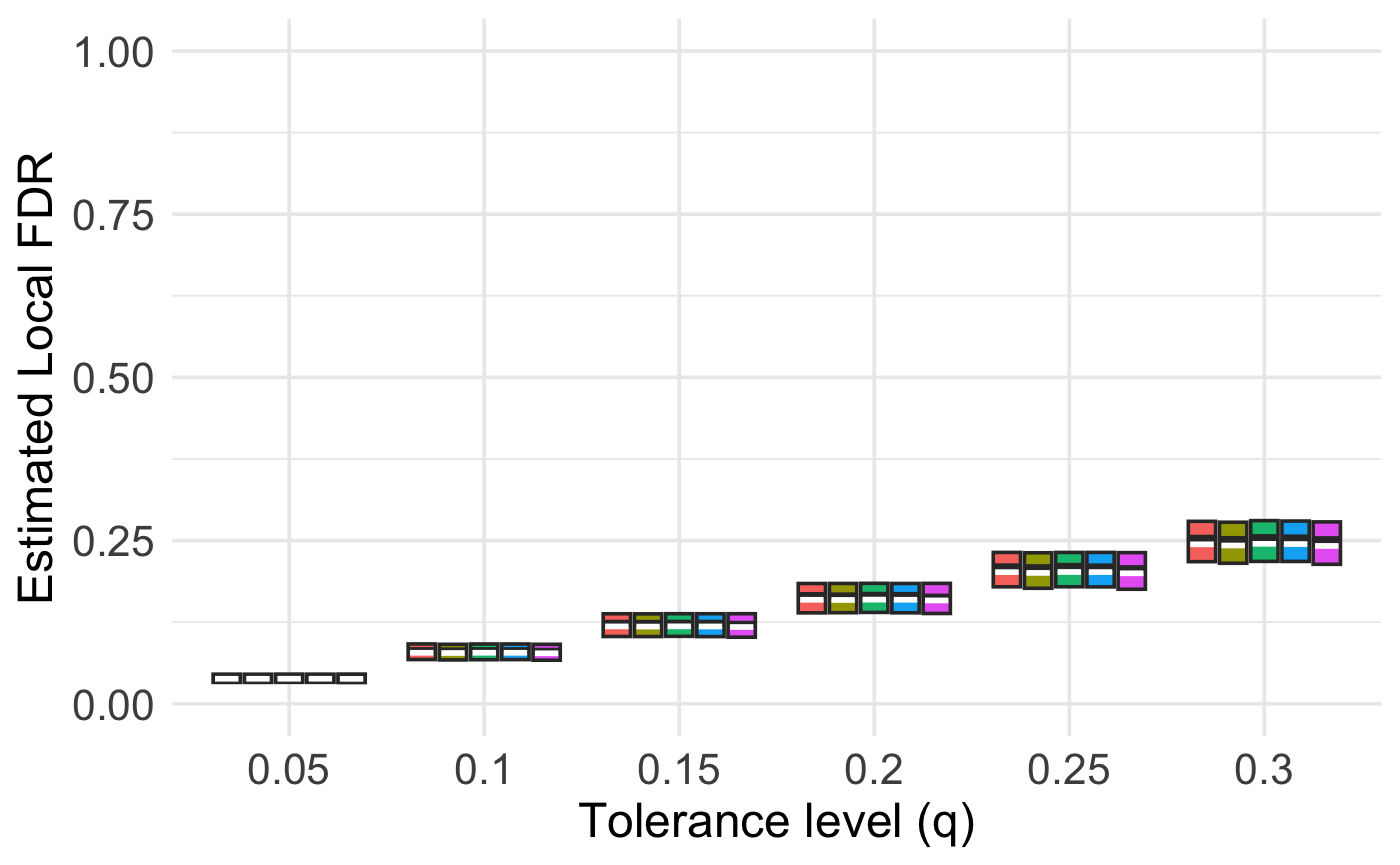}
    \caption{$\widehat{\text{lfdr}}(p_{(R)})$ ($m=1024$)}
\end{subfigure}%

\caption{Shown above are the interquartile ranges of the true lfdr at the estimated cutoff, the estimated lfdr evaluated at the oracle cutoff, and the estimated lfdr evaluated at the estimated cutoff (left to right) for independent test statistics generated 
under an alternating configuration with $\pi_0 = 0.75$ and $N = 10{,}000$.
Black lines denote medians and white lines indicate means.
TSSL is run at level $q$, Storey uses fixed $\lambda = 0.5$, 
and Adaptive Storey uses $\delta = 0.1$.
}
\label{fig:true-est-lfdr-comparison}

\end{figure}

\cite{soloff2024edge} show the Support Line procedure is equivalent to thresholding a monotone estimate of the lfdr, which is useful for assessing significance throughout the rejection set. 
One concern 
is that these lfdr estimates could be quite variable, particularly for smaller values of $m$, since they rely on nonparametric density estimates. If there are too few \textit{p}-values near the cutoff, the estimate $\hat f$ could have such high variance so as to not be useful in practice. We study this variation numerically by examining box plots 
of the estimated lfdr at the $p$-value cutoffs for various choices of $q$, along with the frequentist lfdr function, which is defined for the alternating mean configuration as:
$$\lfdr(t) = \frac{\pi_0}{\pi_0 + (1 - \pi_0)\frac{\frac{1}{4}\sum_{j = 1}^4 \phi(\Phi^{-1}(1 - t) - \frac{5j}{4})}{\phi(\Phi^{-1} (1 - t)) }},$$
evaluated at $t$ equal to the cutoff \textit{p}-value for each procedure. The estimated lfdr is given by $\widehat{\lfdr}(t) = \hat \pi_0 / \hat f(t)$, where $\hat f$ is the non-parametric MLE of a non-increasing density used in the R package \texttt{fdrtool} \citep{strimmer2008fdrtool,grenander1956theory}. Figure \ref{fig:true-est-lfdr-comparison} displays the results for the experiment when $m=64$ and $m=1024$. 

While the estimated lfdr is $\leq \tol$ at the estimated cutoff by construction, the true lfdr has more variation, often exceeding $\tol$ while on average being no larger than $\tol$. Panels (a) and (d) in Figure \ref{fig:true-est-lfdr-comparison} illustrate that the variability in the true lfdr (evaluated at the rejection boundary) increases with the tuning parameter $\tol$ and decreases for larger $m$. This reflects uncertainty around the estimated cutoff, 
particularly when the number of test statistics is small and the Grenander density estimate is more variable. Panels (b) and (e) show that lfdr estimates evaluated at the oracle threshold $t^*$ (for which $\lfdr(t^*)=\tol$) are conservative in this setting, but less so for larger $m$.

\section{Applications}
\label{sec:applications}

Two examples from recent behavioral psychology meta-analyses are analyzed below using the newly proposed procedures. Each example focuses on assessing the effectiveness of psychological interventions: the first involves 122 studies on growth mindset interventions, and the second focuses on 261 `nudge' experiments from behavioral psychology. Figure \ref{fig:histograms} displays histograms of both datasets and Table \ref{tab:combined_num_rejs} lists the number of rejections for each procedure studied here for both datasets. 

\begin{figure}[h!]
    \centering
    \begin{subfigure}[b]{0.45\textwidth}
        \centering
    \includegraphics[width=\linewidth]{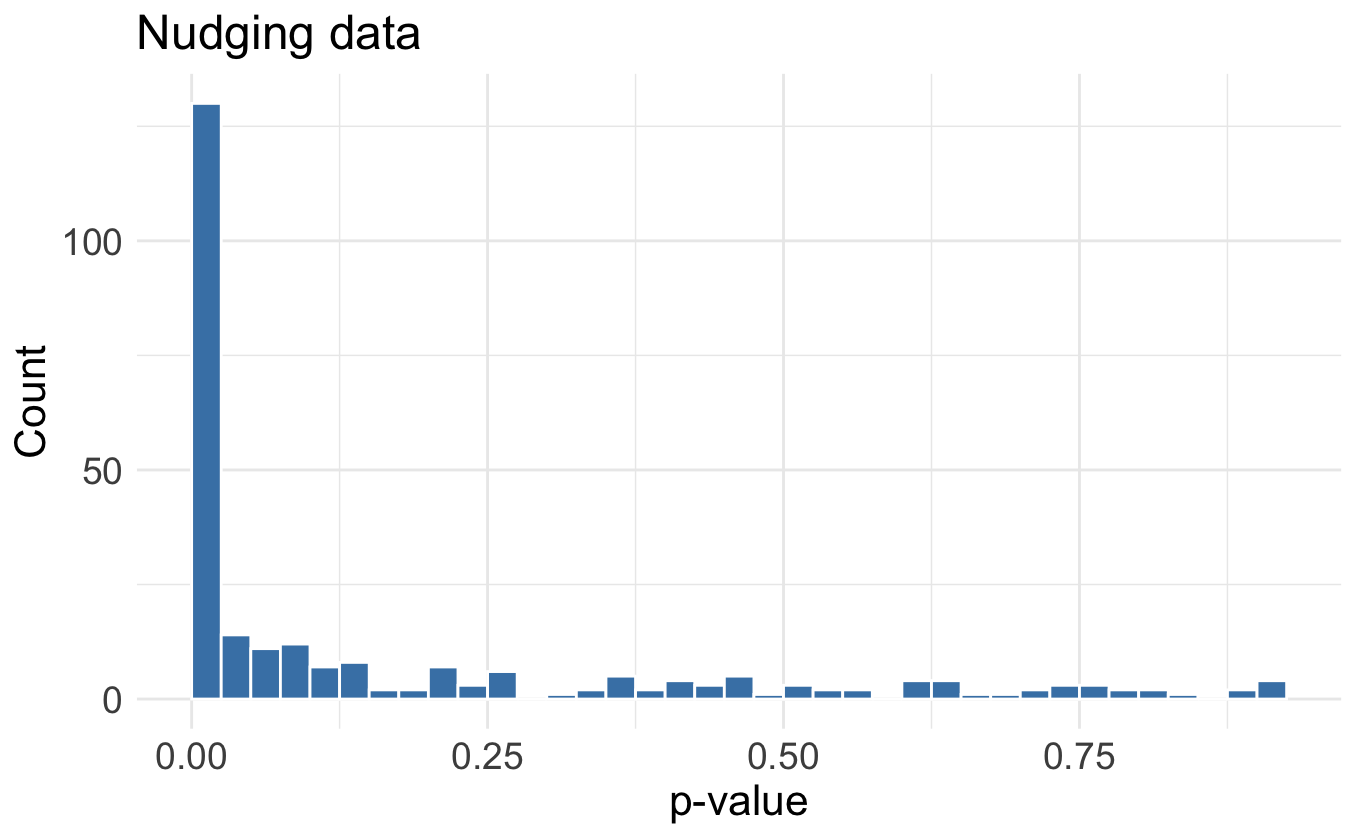}
        \label{fig:merten hist}
    \end{subfigure}
    \hfill
    \begin{subfigure}[b]{0.45\textwidth}
        \centering
        \includegraphics[width=\linewidth]{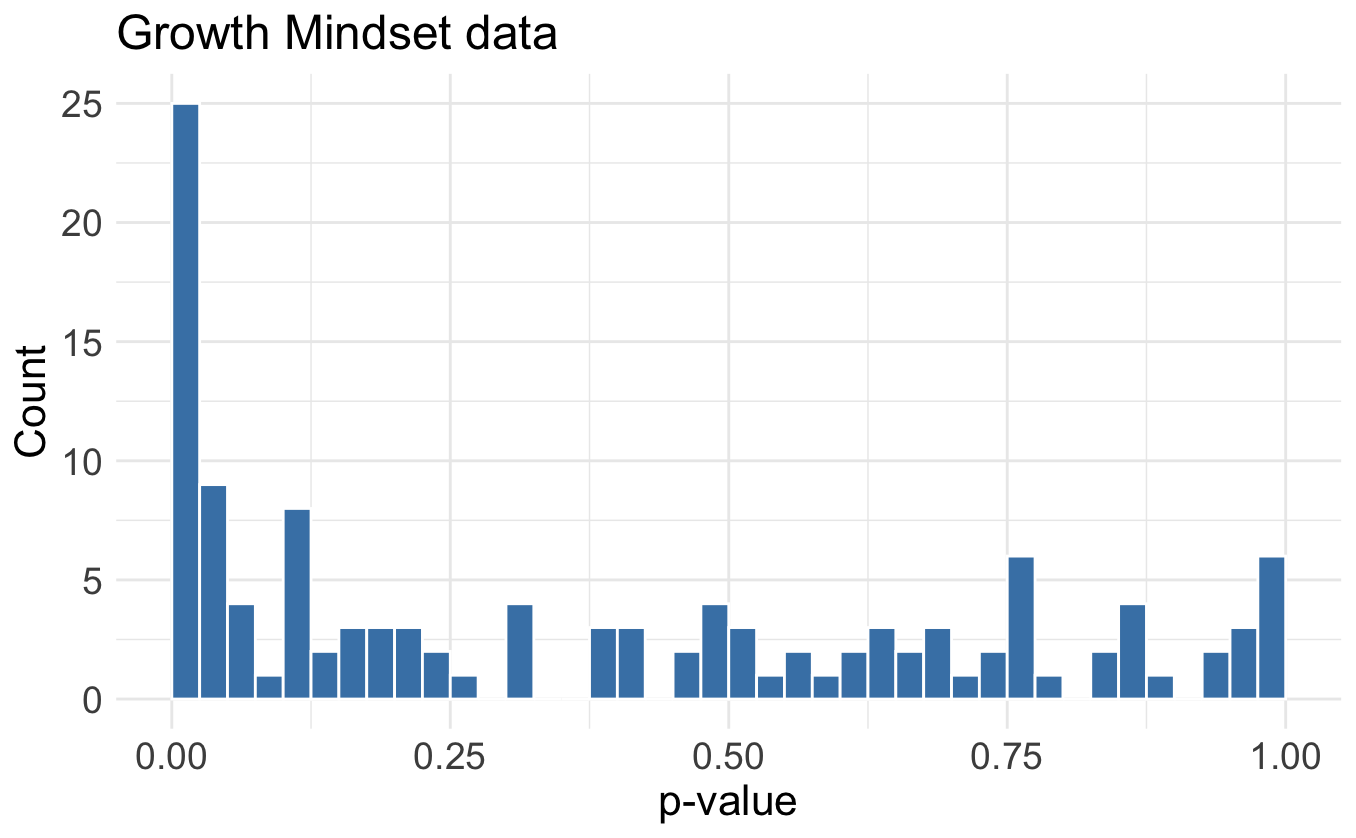}
        \label{fig:macnamara hist}
    \end{subfigure}
    \caption{Histograms of \textit{p}-values for nudging \citep{thaler2009nudge,mertens2022effectiveness} and growth mindset \citep{tipton2023meta,macnamara2023growth} datasets. }
    \label{fig:histograms}
\end{figure}

\subsection{Growth mindset interventions}

There has been debate among psychologists in recent years about the effectiveness of growth mindset interventions, which aim to instill the belief that intelligence is malleable and can be improved with effort. 
\cite{macnamara2023growth} conducted a meta-analysis on 122 studies and concluded that ``apparent effects of growth mindset interventions on
academic achievement are likely attributable to inadequate study design, reporting flaws, and
bias'', finding an insignificant overall effect among studies most closely following best practices $(p=0.666)$. The meta-analysis includes (among many others) three hypothesis tests from \cite{good2003improving}, who tested growth mindset interventions for counteracting stereotype-effects on test performance among female, minority, and low-income groups of middle school students in Texas. 

\cite{good2003improving} randomly assigned 138 seventh-grade students to one of three experimental conditions (teaching that intelligence is expandable, that seventh-grade difficulties are temporary, or both messages) or to a control condition receiving anti-drug messaging,
resulting in experimental vs. control comparisons for female math scores ($p\approx 0.005$), male math scores ($p\approx 0.041$), and reading scores pooled across gender $(p \approx 0.036$). Because all three comparisons use the same control group, the test statistics are positively dependent. This dependence structure is not unique to \cite{good2003improving}; the meta-dataset includes 61 author groups contributing a total of 122 $p$-values with some studies contributing as many as ten $p$-values \citep{outes2020power}. This suggests that comparisons with a common control group are present throughout the dataset, motivating the use of a procedure that is robust to dependence.

Both the BH correction (with $\tol=0.1$) and SL correction (with $\tol=0.2$) discover the female effect, which is consistent with \cite{good2003improving}'s hypothesis that growth mindset interventions improve test scores by counteracting gender or race based stereotypes about mathematical ability, which specifically target females.
BH (with $\tol=0.1$) makes 20 discoveries, while Storey-BH (with $\tol=0.1,\lambda=1/2$) makes 27. By contrast, both SL$(\tol=0.2)$ and TSSL$(\tol'=0.167)$ make only 18 discoveries, hedging against possible bFDR violations due to dependence.

Beyond the dependence concern, Figure \ref{fig:macnamara fdr} shows substantial variation throughout the Storey-BH discovery set; while the overall FDR estimate is under $10\%$, the estimate in the first half $(4\%)$ is less than one third the estimate in the latter half $(14\%)$. Support line methods allow practitioners to assess the reliability of individual interventions they might choose to implement, rather than relying on an aggregate number that masks  this heterogeneity. 
This aligns with the framing by \cite{tipton2023meta}, who challenged the approach of \cite{macnamara2023growth} and advocated for meta-analyses of psychological research to focus instead on revealing where effects are stronger or weaker, rather than seek a single aggregate conclusion as in classical meta-analyses. 

\begin{figure}
    \centering
    \includegraphics[width=0.5\linewidth]{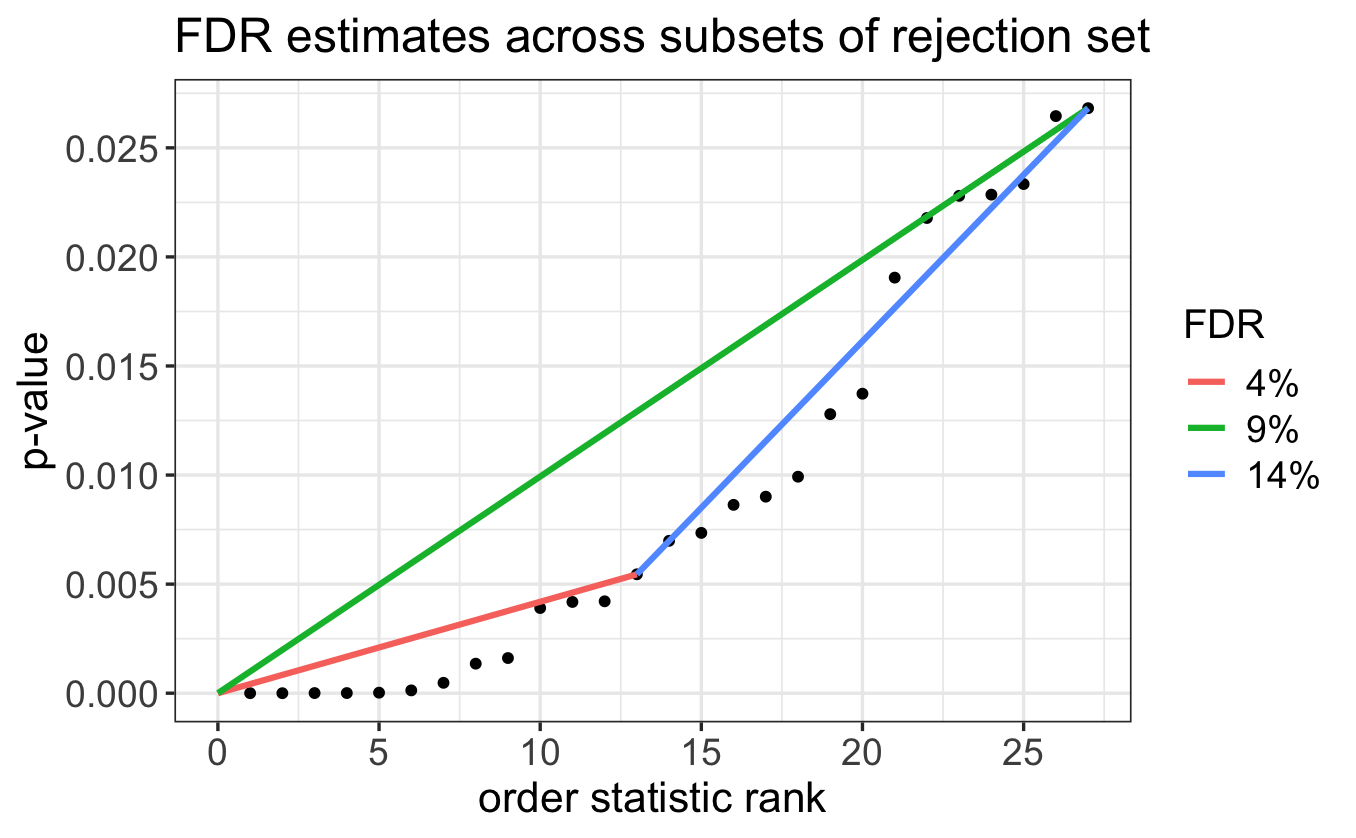}
    \caption{
The figure plots $p$-values against their rank within the Storey-BH rejection set with $\tol=0.1$. The slope of a line connecting the endpoints of
any contiguous subset is proportional to the estimated FDR
within that subset. 
The numbers listed to the right of the graph give $\widehat{\FDR} = \frac{\hat{m}_0 (p_{(t)}-p_{(s)})}{t-s}$, 
where $\hat{m}_0\approx 92$ 
is the Storey estimate of $m_0$ with $\lambda=1/2$. For example, the estimated FDR within the second half of the rejection set (blue) is $92 \times (0.0268-0.005) / (27-13) \approx 14\%$.
    } 
    \label{fig:macnamara fdr}
\end{figure}

\subsection{Nudge interventions}


Another recent meta-analysis mentioned in the critique by \cite{tipton2023meta} is the one by \cite{mertens2022effectiveness}, who compiled data from $447$ nudge experiments, with the goal of assessing the overall effectiveness of nudging.
Broadly, nudge interventions attempt to influence peoples' behavior in a predictable way without restricting their options \citep{thaler2009nudge}. Due to concerns of publication bias \citep{maier2022no}, we restrict attention to the subset of $261$ nudge experiments whose $p$-value was below the 5\% two-sided significance level, following a selection-adjustment procedure  \citep{hung2020statistical}. For more details on this, see Section 6.3 of \cite{xiang2025frequentist}.

Table \ref{tab:combined_num_rejs} shows the number of  rejections for each procedure for the nudging and growth mindset datasets, respectively. We see considerable gains in power for the adaptive procedures compared to the standard Support Line procedure, with the Storey procedures resulting in the highest number of rejections. 
We also see little difference between the various Storey and Adaptive Storey procedures. We note that for the growth-mindset data, TSSL($q'$) and SL have the same number of rejections, which aligns with the observation from Figure \ref{fig:power_heatmaps} that TSSL($q'$) does not demonstrate gains in power over SL when $ \pi_0$ is close to 1.

The nudging data has a higher proportion of near-zero \textit{p}-values than the growth mindset data, which is consistent with the results in Tables \ref{tab:combined_pi0} and \ref{tab:combined_num_rejs}; the nudging data has smaller $\hat \pi_0$ estimates and a higher percentage of rejections. 

\begin{table}[h!]
\centering
\caption{$\hat{\pi}_0$ estimates for Nudging and Growth Mindset datasets }
\label{tab:combined_pi0}
\begin{tabular}{lccc|ccc}
\hline
 & \multicolumn{3}{c|}{\textbf{Nudging}} & \multicolumn{3}{c}{\textbf{Growth mindset}} \\
\textbf{Procedure} & $q=0.1$ & $q=0.2$ & $q=0.3$ & $q=0.1$ & $q=0.2$ & $q=0.3$ \\
\hline
TSSL$(q)$ & 0.62 & 0.56 & 0.51 & 0.93 & 0.85 & 0.78 \\
TSSL$(q')$ & 0.62 & 0.56 & 0.56 & 0.93 & 0.85 & 0.78 \\
Storey$(1/2)$ & 0.28 & 0.28 & 0.28 & 0.75 & 0.75 & 0.75 \\
Storey$(q)$ & 0.40 & 0.36 & 0.33 & 0.77 & 0.70 & 0.73 \\
LSL & 0.39 & 0.39 & 0.39 & 0.77 & 0.77 & 0.77 \\
AS$(0.01)$ & 0.38 & 0.35 & 0.33 & 0.72 & 0.70 & 0.72 \\
AS$(0.1, 0.5)$ & 0.29 & 0.29 & 0.29 & 0.80 & 0.80 & 0.80 \\
AS$(0.1)$ & 0.29 & 0.29 & 0.29 & 0.73 & 0.73 & 0.75 \\
\hline
\end{tabular}
\end{table}

\begin{table}[ht]
\centering
\caption{Number of rejections for each procedure (and percentage of total hypotheses rejected) with tuning parameters $\tol$ and $\tol'=\frac{\tol}{1+\tol}$.}
\label{tab:combined_num_rejs}

\begin{tabular}{lccc|ccc}
\hline
 &\multicolumn{3}{c|}{\textbf{Nudging}} & \multicolumn{3}{c}{\textbf{Growth mindset}} \\
\textbf{Procedure} &$q=0.1$ &  $q=0.2$ & $q=0.3$ & $q=0.1$ & $q=0.2$ & $q=0.3$ \\
\hline
SL$(\tol)$ &99 (38\%) & 115 (44\%) & 129 (49\%) & 9 (7\%) & 18 (15\%) & 27 (22\%) \\
TSSL$(q)$ &115 (44\%) & 129 (49\%) & 162 (62\%) & 12 (10\%) & 27 (22\%) & 34 (28\%) \\
TSSL$(q')$ &115 (44\%) & 129 (49\%) & 131 (50\%) & 9 (7\%) & 18 (15\%) & 27 (22\%) \\
Storey$(1/2)$ &129 (49\%) & 162 (62\%) & 182 (70\%) & 18 (15\%) & 27 (22\%) & 34 (28\%) \\
Storey$(q)$ &129 (49\%) & 162 (62\%) & 174 (67\%) & 18 (15\%) & 27 (22\%) & 34 (28\%) \\
LSL &129 (49\%) & 162 (62\%) & 162 (62\%) & 18 (15\%) & 27 (22\%) & 34 (28\%) \\
AS$(0.01)$ &129 (49\%) & 162 (62\%) & 174 (67\%) & 18 (15\%) & 27 (22\%) & 34 (28\%) \\
AS$(0.1, 0.5)$ &129 (49\%) & 162 (62\%) & 182 (70\%) & 18 (15\%) & 27 (22\%) & 34 (28\%) \\
AS$(0.1)$ &129 (49\%) & 162 (62\%) & 182 (70\%) & 18 (15\%) & 27 (22\%) & 34 (28\%) \\
\hline
\end{tabular}
\end{table}

\begin{table}
\centering
\caption{Adaptive BH procedures on Growth Mindset Data using tolerance level $q = 0.1$ and $q' = 0.91$}
\label{tab:macnamara bh}
\begin{tabular}{>{\raggedright\arraybackslash}p{0.15\linewidth}>{\raggedright\arraybackslash}p{0.15\linewidth}>{\raggedright\arraybackslash}p{0.15\linewidth}>{\raggedright\arraybackslash}p{0.15\linewidth}}
 & & &\\
 \textbf{Procedure} & Number of rejections& FDR of first half of rejections& FDR of second half of rejections\\
\hline

BH$(\tol)$& 20 & 0.0476 & 0.1198 \\
TST$(q)$& 25 & 0.0358 & 0.1501 \\
TST$(q')$& 20 & 0.0398 & 0.1002 \\
Storey$(1/2)$ & 27 & 0.0386 & 0.1404 \\
Storey$(q)$ & 27 & 0.0391 & 0.1425 \\
LSL & 27 & 0.0396 & 0.1443 \\
AS$(0.01)$ & 28 & 0.0436 & 0.1553 \\
AS$(0.1, 0.5)$ & 27 & 0.0409 & 0.1488 \\
AS$(0.1)$ & 27 & 0.0371 & 0.1352 \\

\end{tabular}

\end{table}

Table \ref{tab:macnamara bh} shows the results of the Benjamini-Hochberg analogues to the adaptive SL procedures. We ran the BH procedures at level $\tol = 0.1$ because BH is less conservative than the support line procedure. We see that these BH procedures reject more \textit{p}-values than the SL procedures at level $\tol = 0.1$ and a similar number to the SL procedures at level $\ell = 0.2$. We can also observe that the \textit{p}-values in the rejection set are quite heterogeneous; the FDR in the second half of the rejection set for each procedure is more than double that of the first half.   


\section{Discussion}

The adaptive bFDR procedures studied here achieve gains in power over the standard SL procedure as demonstrated in both simulations and on real data sets. Our experiments suggest the two-stage SL procedure and Storey $(\lambda=0.2)$ work well under positive dependence. In the next few paragraphs, we outline connections with existing work and future directions. 

\begin{figure}[h!]
    \centering

    \begin{minipage}{0.4\textwidth}
        \centering
        \includegraphics[width=\linewidth]{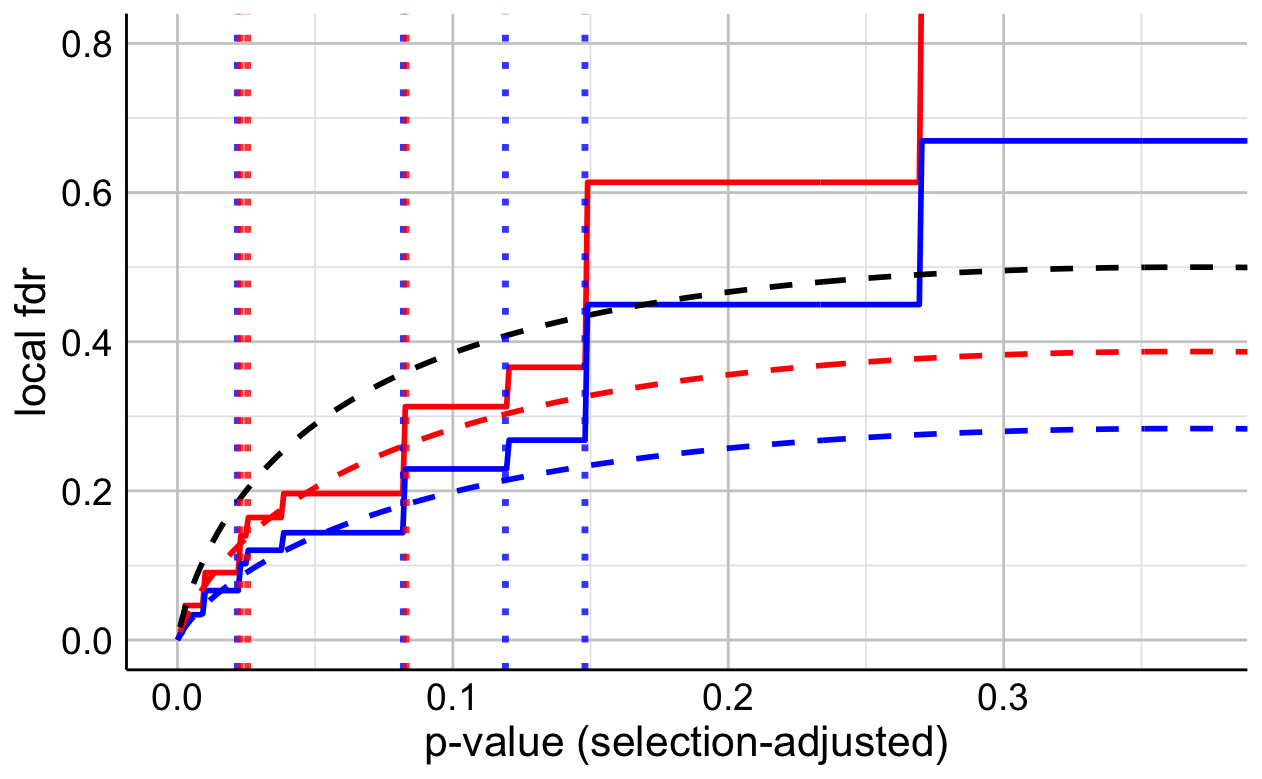}
        \caption*{\footnotesize Psychological nudging \textit{p}-values (selection-adjusted) \citep{mertens2022effectiveness}}
        \label{fig:merten-alpha}
    \end{minipage}
    \hfill
    \begin{minipage}{0.4\textwidth}
        \centering
        \includegraphics[width=\linewidth]{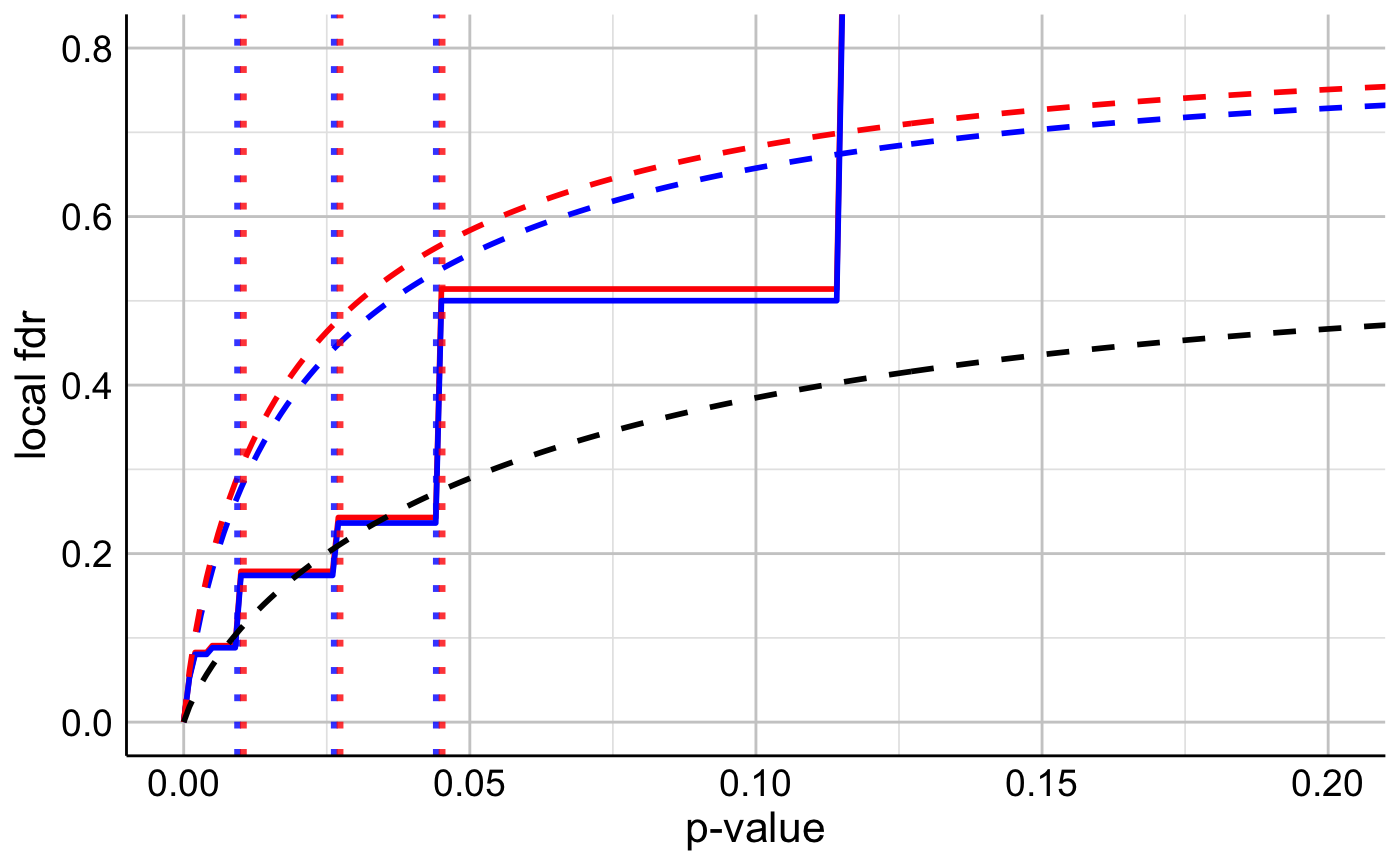}
        \caption*{\footnotesize Growth mindset \textit{p}-values \citep{macnamara2023growth}}
        \label{fig:macnamara-alpha}
    \end{minipage}
    \hfill
    \begin{minipage}{0.18\textwidth}  
        \centering
        \vspace*{-0.8cm}
        \includegraphics[width=\linewidth]{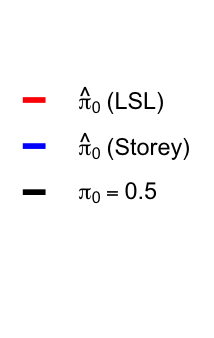}
        \label{fig:legend-lfdr-alpha}
    \end{minipage}

    \caption{$\alpha_{\hat \pi_0}(p)$ (dashed lines) and $\widehat{\lfdr}_{\hat \pi_0}(p)$ (solid lines) with vertical lines at \textit{p}-value cutoffs for $\tol \in \{0.1, 0.15, 0.2, 0.25, 0.3\}$. }
    \label{fig:lfdr-and-alpha}
\end{figure}

\paragraph{Dependence.}
More work is needed to establish bFDR guarantees under dependence and to identify whether certain forms of dependence can preserve the finite-sample bFDR control for the Support Line procedure and its adaptive variants.
Figure \ref{fig:corr bfdr} illustrates the importance of using procedures that maintain bFDR control under dependence,
since some procedures that had high power under independence also exhibited the poor performance in terms of bFDR control under correlation. 
More effort is needed to explore performance of adaptive SL procedures under a broader class of dependence structures beyond the equicorrelated setting considered in this paper. Additionally, since Adaptive Storey with the stopping rule described in \ref{sec:main} had drastically different bFDR under equicorrelation for different choices of $\delta$, it would be useful to develop guidelines for choosing $\delta$ under both independence and dependence assumptions. 



\paragraph{Calibration of $p$-values.}

\cite{sellke2012calibration} define a calibration function for $p$-values that can be interpreted as a lower bound on posterior probability of Type I error when the non-null $p$-values are drawn from a Beta$(\xi, 1)$ distribution. Assuming equal prior odds for the truth or falsehood of each null hypothesis, the local fdr for a $p$-value equal to $t$ is lower bounded by
\begin{align}
\label{eq:sellke-fun}
    \alpha(t) &= \frac{t\log (1/t)}{e^{-1}+t\log (1/t)}, 
    \hspace{2em} t<1/e.
\end{align}
If instead of assuming equal prior odds, the estimate $\hat \pi_0$ is plugged in for $\pi_0$, one arrives at:
\begin{align}
\label{eq:sellke-fun-adaptive}
    \alpha_{\hat \pi_0}(t) = \frac{t \log(1/t)}{e^{-1}\left( \frac{1 - \hat \pi_0}{\hat \pi_0} \right)+t \log(1/t) }, 
    \hspace{2em} t<1/e.
\end{align}
Figure \ref{fig:lfdr-and-alpha} shows plots of $\alpha(t),\alpha_{\hat{\pi}_0}(t)$ and $\widehat{\lfdr}_{\hat \pi_0}(t)$ using Storey ($\lambda=1/2$) and LSL estimates for $\hat \pi_0$. Vertical lines are drawn at the common cutoff values $\tol = 0.1, 0.15, 0.2, 0.25, 0.3$. For the growth mindset $p$-values, the local fdr estimates roughly agree with the calibration method of \cite{sellke2012calibration}. For the nudge data, $\alpha(t)$ is more conservative than $\widehat{\lfdr}(t)$. This is likely because the $\hat \pi_0$ estimates given by the Storey and LSL procedures are smaller than 0.5 (see Table \ref{tab:combined_pi0}), which is the implicit value used in the derivation of \eqref{eq:sellke-fun}. For the growth mindset dataset, where $\pi_0$ is estimated to be larger, the adaptive version \eqref{eq:sellke-fun-adaptive} is more conservative. The \textit{p}-value cutoffs should not be directly compared between the two datasets, as the nudging \textit{p}-values are selection-adjusted\footnote{The \textit{p}-values are adjusted for publication bias by restricting the set of one-sided \textit{p}-values to those below 0.025 and then rescaling by a factor of 40.} while the growth mindset \textit{p}-values are not. We chose not to adjust the growth mindset \textit{p}-values because \cite{tipton2023meta} found no evidence of publication bias after accounting for heterogeneity. 


\paragraph{Dominance of adaptive procedures.}
Numerically, we have found that using \textit{some} estimator for $\pi_0$ is better than using none when $m_0$ is not too close to $m$, in that there is substantial power to be gained while maintaining provable boundary FDR control. However, theoretical guarantees on TSSL and Storey SL in \ref{thm:bfdr-main} and \ref{thm:AS-bFDR-control} are both inequalities, which allows for the possibility of matching or even under-performing relative to the standard SL procedure, particularly if $\pi_0$ is close to 1. We see examples of this in Figure \ref{fig:power_heatmaps} and Table \ref{tab:combined_num_rejs}. The simulation results in Section \ref{sec:type1-sim} suggest these upper bounds can be conservative, while Theorem \ref{thm:asymptotic-lfdr-control} (see Section \ref{sec:appendix-asymptotic}) characterizes the gap between the realized boundary lfdr and the upper bound $\tol/(1-\tol)$ when the number of tests grows large, notably without requiring independence. An avenue for future research is to investigate when adaptive bFDR procedures provide a definitive theoretical advantage over the original SL procedure.

\section*{Acknowledgements}

We thank Yoav Benjamini, William Fithian, Zijun Gao, Nikolaos Ignatiadis, Etienne Roquain and Jake Soloff for thoughtful discussions related to this work. \\

\noindent\textbf{Reproducibility.} Code to reproduce all figures is available on Github, at the repository: \href{https://github.com/sarahmostow/sarahmostow.github.io/tree/main/Adaptive%20Support%20Line}{github.com/sarahmostow/sarahmostow.github.io/tree/main/Adaptive\%20Support\%20Line}

\section{Proofs}
\label{sec:proofs}

\textbf{Theorem \ref{thm:bfdr-main}.} 
    Let $H_1,\dots,H_m$ denote $m$ null hypotheses, with independent $p$-values $p_1,\dots,p_m$. Suppose that $p_i \sim \text{Uniform}(0,1)$ if $H_i$ is true. Then
    \begin{align*}
        \bFDR(\mathcal{R}_2) \leq \frac{\tol}{1-\tol},
    \end{align*}
where $\mathcal{R}_2 \coloneqq \{i : p_i \leq p_{(R_2)}\}$ is the rejection set for the two stage procedure.

\begin{proof}
We decompose the bFDR into two pieces:
\begin{align}
\label{eq:bFDR-decomposition}
    \P(H_{(R_2)}\text{ is true}, R_2 > 0) = \P(H_{(R_2)}\text{ is true},R_1=R_2) + \P(H_{(R_2)}\text{ is true},R_1 < R_2),
\end{align}
with the convention that $H_{(0)}\coloneqq$ false. For the first piece in \eqref{eq:bFDR-decomposition}, note by exchangeability
\begin{align*}
    \P(H_{(R_2)}\text{ is true},R_1=R_2) = m_0 \P(p_{(R_1)}=p_m,R_1=R_2),
\end{align*}
where we assumed wlog that $H_m$ is true. Let $q_{(1)}\leq \dots \leq q_{(m-1)}$ denote the order statistics of $p_1,\dots,p_{m-1}$, and define
\begin{align*}
    \Delta_k \coloneqq \frac{\tol k}{m}-q_{(k)}, \hspace{2em}k=0,\dots,m,
\end{align*}
where $q_{(m)}\coloneqq 1$ and $q_{(0)}\coloneqq 0$. On the event $\{p_{(R_1)}=p_m\}$, we have for some $k\in \{0,\dots,m-1\}$ that $q_{(k)} < p_m < q_{(k+1)}$ and
\begin{align*}
    \frac{\tol (k+1)}{m}-p_m \geq \big(\max_{j \leq k} \Delta_j\big) \vee \Big(\max_{j>k} \Delta_j + \frac{\tol}{m} \Big).
\end{align*}
When non-empty, the intersection of these two constraints is an interval $I_k$, of the form:
\begin{align}
\label{eq:interval-k}
    I_k \coloneqq \Big(q_{(k)}, \; \tol k/m- (\max_{j > k}\Delta_j) \vee \big(\max_{j\leq k}\Delta_j - \tol/m \big)  \Big),
\end{align}
where we define $I_k \coloneqq \varnothing$ if the right endpoint of \eqref{eq:interval-k} is smaller than the left endpoint. It follows from disjointness of $I_0,\dots,I_{m-1}$ that
\begin{align}
\label{eq:first-piece}
    \P(p_{(R_2)}=p_m,R_1=R_2 \mid p_{-m}) \leq \P(\cup_{k=0}^{m-1}\; \{p_m \in I_k\}\mid p_{-m}) = \sum_{k=0}^{m-1}\P(p_m \in I_k\mid p_{-m}).
\end{align}
Next, we consider the case $R_1 < R_2$, which in particular implies $R_1 < m$, and define:
\begin{align*}
    \Delta_k' \coloneqq \frac{\tol k}{m-R_1(1)} - q_{(k)}, \hspace{2em} k=0,\dots,m,
\end{align*}
where $R_1(p_m)$ denotes the value of $R_1$ for a given value of $p_m$, holding $p_1,\dots,p_{m-1}$ fixed (see \ref{lem:p-to-one} for definition of $R_1( \cdot )$). Note that on the event $\{R_1 < R_2\} \cap\{p_{(R_2)}=p_m\}$ we have by Lemma \ref{lem:p-to-one} that $R_1(p_m)=R_1(1)$.
For the second piece in \eqref{eq:bFDR-decomposition}, it follows that on the event $\{R_1 < R_2\} \cap\{p_{(R_2)}=p_m\}$, we have for some $k$ that $q_{(k)}<p_m < q_{(k+1)}$ and 
\begin{align*}
    \frac{\tol(k+1)}{m-R_1(1)}- p_m &\geq \big(\max_{j \leq k} \Delta_j'\big) \vee \Big(\max_{j>k} \Delta_j' + \frac{\tol}{m-R_1(1)} \Big) \\
    \frac{\tol (k+1)}{m}-p_m &< \big(\max_{j \leq k} \Delta_j\big) \vee \Big(\max_{j>k} \Delta_j + \frac{\tol}{m} \Big),
\end{align*}
which rearrange to yield the following constraints: for some $k$, 
\begin{align}
\label{ineq:lb1}
q_{(k)} <p_m &<q_{(k+1)} \\
\nonumber
    p_m &\leq \frac{\tol k}{m-R_1(1)} - \big( \max_{j > k}\Delta_j' \big) \vee \Big( \max_{j \leq k} \Delta_j' - \frac{\tol}{m-R_1(1)} \Big)\\
    \label{ineq:lb2}
    p_m &> \frac{\tol k}{m}- \big( \max_{j> k} \Delta_j\big) \vee \big( \max_{j\leq k}\Delta_j - \tol/m  \big).
\end{align}
For each $k$, we define the analogous interval to $I_k$ as 
\begin{align*}
    I_k' \coloneqq \Big(q_{(k)}, \; \frac{\tol k}{m-R_1(1)} - \big( \max_{j > k}\Delta_j' \big) \vee \Big( \max_{j \leq k} \Delta_j' - \frac{\tol}{m-R_1(1)} \Big) \Big),
\end{align*}
where $I_k'\coloneqq \varnothing$ if the right endpoint of the above is smaller than the left endpoint. The two lower constraints \eqref{ineq:lb1} and \eqref{ineq:lb2} for the value of $p_m$ prompt us to consider two cases. 

\paragraph{Case 1:} $q_{(k)} > \frac{\tol k}{m}- \big( \max_{j> k} \Delta_j\big) \vee \big( \max_{j\leq k}\Delta_j - \tol/m  \big)$. \\

\noindent In this case, $I_k$ as defined in \eqref{eq:interval-k} is empty, and $I_k'$ is the set over which $p_m$ achieves the maximum (in the second stage) as the $(k+1)^{\th}$ order statistic, which has length:
\begin{align*}
    \Delta_{k}' - \big( \max_{j > k}\Delta_j' \big) \vee \Big( \max_{j \leq k} \Delta_j' - \frac{\tol}{m-R_1(1)} \Big).
\end{align*}

\paragraph{Case 2:} $q_{(k)} \leq \frac{\tol k}{m}- \big( \max_{j> k} \Delta_j\big) \vee \big( \max_{j\leq k}\Delta_j - \tol/m  \big)$. \\

\noindent In this case, $I_k$ is non-empty and the length of the set over which $p_m$ achieves the maximum (in the second stage) as the $(k+1)^{\th}$ order statistic is:
\begin{align*}
    \frac{\tol k}{m-R_1(1)} - \big( \max_{j > k}\Delta_j' \big) \vee \Big( \max_{j \leq k} \Delta_j' - \frac{\tol}{m-R_1(1)} \Big) - \left[\frac{\tol k}{m}- \big( \max_{j> k} \Delta_j\big) \vee \big( \max_{j\leq k}\Delta_j - \tol/m  \big) \right].
\end{align*}
Adding and subtracting $q_{(k)}$, we re-express the above length as:
\begin{align*}
    \underbrace{\Delta_k' - \big( \max_{j > k}\Delta_j' \big) \vee \Big( \max_{j \leq k} \Delta_j' - \frac{\tol}{m-R_1(1)} \Big)}_{\P(p_m \in I_k' \mid p_{-m})} - \underbrace{\big[\Delta_k - \big( \max_{j> k} \Delta_j\big) \vee \big( \max_{j\leq k}\Delta_j - \tol/m  \big) \big]}_{\P(p_m \in I_k \mid p_{-m})},
\end{align*}
since $p_m \mid p_{-m} \sim \text{Uniform}(0,1)$, where $p_{-m}\coloneqq (p_1,\dots,p_{m-1})$. In summary, we have shown
\begin{align*}
    \P(p_{(R_2)}=p_m,R_1=R_2 \mid p_{-m}) &\leq \sum_{k=0}^{m-1}\P(p_m \in I_k\mid p_{-m}) \\
    \P(p_{(R_2)}=p_m,R_1<R_2 \mid p_{-m}) &\leq \sum_{k=0}^{m-1} \P(p_m \in I_k'\mid p_{-m}) - \P(p_m \in I_k\mid p_{-m}).
\end{align*}
The sum is therefore bounded by
\begin{align*}
\P(p_{(R_2)}=p_m\mid p_{-m}) &= \P(p_{(R_2)}=p_m,R_1=R_2\mid p_{-m})+\P(p_{(R_2)}=p_m,R_1<R_2\mid p_{-m}) \\
    &\leq  \sum_{k=0}^{m-1} (\P(p_m \in I_k'\mid p_{-m}) - \P(p_m \in I_k\mid p_{-m})) + \P(p_m \in I_k\mid p_{-m}) \\
    &= \sum_{k=0}^{m-1} \Delta_k' - \big( \max_{j > k}\Delta_j' \big) \vee \Big( \max_{j \leq k} \Delta_j' - \frac{\tol}{m-R_1(1)} \Big),
\end{align*}
which is telescoping and equal to $\frac{\tol}{m-R_1(1)}$. Taking expectation on both sides, we have shown
\begin{align*}
    \P(H_{(R_2)}\text{ is true},R_2>0) \leq \E \Big[ \frac{m_0 \tol}{m-R_1(1)} \Big].
\end{align*}
We show in Lemma \ref{lem:expectation-bound} that the right hand side of the above is $\leq \frac{\tol}{1-\tol}$, completing the proof.
\end{proof}

\noindent \textbf{Proposition \ref{thm:AS-bFDR-control}.}
    Let $H_1,\dots,H_m$ be null hypotheses, and suppose that for any $i$ where $H_i$ is true, we have that $p_i \sim \text{Uniform}(0,1)$ and that $p_i$ is independent of $p_{-i}:=(p_{j}:j \in [m]\backslash \{i\})$. Then
\begin{align*}
    \bFDR(\mathcal{R}_{\tol/\hat{\pi}_0}) \leq \tol.
\end{align*}

\begin{proof}
    Suppose wlog $H_m$ is true. By exchangeability of the null $p$-values, and the tower property, we have
    \begin{align}
    \label{eq:AS-piece1}
        \bFDR = m_0 \P(p_{(R_{\tol/\hat{\pi}_0})}=p_m ) = m_0 \E\big[ \P(p_{(R_{\tol/\hat{\pi}_0})}=p_m \mid p_{-m},1\{p_m \leq \tol\} ) \big].
    \end{align}
    Define $p_{i}' := p_i / \tol$, $m' := \# \{i: p_i \leq \tol\}$, and $\tol' := \frac{m'}{m\hat{\pi}_0}$. Then
    \begin{align*}
        R_{\tol/\hat{\pi}_0} &:= \argmin_{k:p_{(k)}\leq \tol} \Big\{ p_{(k)} - \frac{\tol k}{\hat{\pi}_0 m} \Big\} \\
        &= \argmin_{k:p_{(k)}\leq \tol} \Big\{ \frac{p_{(k)}}{\tol} - \frac{k}{\hat{\pi}_0 m} \Big\} = \argmin_{k=0,1,\dots,m'} \Big\{ p_{(k)}' - \frac{\tol' k}{m'} \Big\}.
    \end{align*}
    Notice $\{p_{(R_{\tol/\hat{\pi}_0})}=p_m\}$ implies $\{p_m \leq \tol\}$. It follows from Lemma 2 of \cite{soloff2024edge} (restated in Lemma \ref{lem:SL-key-lemma} below) 
    that
    \begin{align*}
        \P(p_{(R_{\tol/\hat{\pi}_0})}=p_m \mid p_{-m},1\{p_m \leq \tol\} ) \leq \frac{\tol'\; 1\{p_m \leq \tol\}}{m'} = \frac{1\{p_m \leq \tol\}}{m\hat{\pi}_0} = \frac{1\{p_m \leq \tol\}}{(m-1) \hat{\pi}_0(p_{-m})}
    \end{align*}
    since $\hat{\lambda}\geq \tol$ implies $\hat{\pi}_0(p_1,\dots,p_m) = \frac{m-1}{m} \; \hat{\pi}_0(p_1,\dots,p_{m-1})$ on the event $\{p_m \leq \tol\}$. Taking expectation on both sides while retaining the conditioning on $p_{-m}$, and noting that $p_m \sim \text{Uniform}(0,1)$ independently of $p_{-m}$, we have
    \begin{align*}
        \P(p_{(R_{\tol/\hat{\pi}_0})}=p_m \mid p_{-m}) \leq \frac{1}{(m-1)\hat{\pi}_0 (p_{-m})} \; \E\big[ 1\{p_m \leq \tol\} \mid p_{-m}\big] = \frac{\tol}{(m-1) \hat{\pi}_0 (p_{-m})}.
    \end{align*}
    We were able to pull out $1/\hat{\pi}_0(p_{-m})$ from the expectation because $\hat{\lambda}$ is a function of $p_{-m}$ on the event $\{p_m \leq \tol\}$, and thus $\hat{\pi}_0$ is determined by $p_{-m}$. Next, notice that
    \begin{align*}
        \frac{\tol}{(m-1)} \; \E\Big[\frac{1}{\hat{\pi}_0(p_{-m})} \Big] &= \frac{\tol}{(m-1)} \; \E\Big[\frac{(m-1)(1-\hat{\lambda})}{1+\# \{i < m: p_i > \hat{\lambda}\}} \Big]  \\
        &\leq \tol \; \E\Big[\frac{1-\hat{\lambda}}{1+\# \{i < m: H_i \text{ is true}, \; p_i > \hat{\lambda}\}} \Big] .
    \end{align*}
    Since $\hat{\lambda}\in [\tol,1]$ is a stopping time, and the expression inside the expectation is a super-martingale (see proof of Theorem 1 in \cite{gao2025adaptive}) the OST implies that the above is bounded by
    \begin{align}
    \label{eq:AS-piece2}
        &\leq \tol \; \E \Big[ \frac{1-\tol}{1+\# \{i<m:H_i\text{ is true},\; p_i > \tol\}} \Big] \leq \tol \; \frac{1-\tol}{ (1+(m_0-1)) (1-\tol)}
    \end{align}
    following from standard Binomial calculations (see, e.g. the proof of Theorem 1 in \cite{gao2025adaptive}). Together, \eqref{eq:AS-piece1} and \eqref{eq:AS-piece2} imply
    \begin{align*}
        \bFDR = m_0 \P(p_{(R_{\tol/\hat{\pi}_0})}=p_m) \leq m_0 \; \frac{\tol (1-\tol)}{m_0 (1-\tol)} = \tol,
    \end{align*}
    completing the proof.
\end{proof}

\subsection{Technical Lemmas}

\begin{lemma}
\label{lem:p-to-one}
    Let $P_{(-m)}(t) \in [0,1]^m$ denote the vector of p-values obtained by replacing $p_m$ with $t$, and let $R_1(t): [0, 1]^m \rightarrow \{0, ..., m\}$ denote the number of rejections made in stage one of the TSSL procedure run on the vector $P_{(-m)}(t)$. If $p_m > p_{(R_1)}$, then $R_1(p_m) = R_1(1)$. 
\end{lemma}
\begin{proof}
    Suppose $p_m$ is the $(k+1)^{\th}$ order statistic and consider what happens to the number of rejections in Stage 1 if we replace $p_m$ with $1$. By assumption $R_1(p_m) \leq k$, and since moving $p_m$ to 1 shifts the ranks of $p_{(k+2)},\dots,p_{(m)}$ down by 1, we have 
    \begin{align*}
        \Big[\max_{j \in \{0,\dots,k\}\backslash \{R_{1}(p_m)\}} \Delta_j \Big] \vee \Big[\max_{j > k+1} \Delta_j- \frac{\tol}{m} \Big] \leq \max_{j \in \{0,\dots,m\}\backslash \{R_1(p_m)\}}\Delta_j < \Delta_{R_1(p_m)},
    \end{align*}
    which implies $R_1(p_m) = R_1(1)$.
\end{proof}

\begin{lemma}[Adapted from Lemma 1 in \cite{benjamini2006adaptive}]
    \label{lem: expectation-bound-sub-lemma}
    For any $n \geq 2$, if  $Y \sim Binomial(n, p)$, then $\mathbb{E}\bigl[\frac{1}{Y+1} \bigr] \leq \frac{1}{(n+1)p}$. 
    \begin{proof}
        We have 
        \begin{align*}
            \mathbb{E}\left[\frac{1}{Y+1} \right] &= \sum_{k=0}^n \frac{1}{k+1} \binom{n}{k}p^k(1-p)^{n-k} \\
            &= \sum_{k=0}^n \frac{1}{n+1} \binom{n+1}{k+1}p^k(1-p)^{n-k} \\
            &= \frac{1}{p(n+1)} \sum_{j=1}^{n+1}  \binom{n+1}{j}p^j(1-p)^{(n+1) - j} \\
            &= \frac{1}{p(n+1)} \bigg( \sum_{j=0}^{n+1}  \binom{n+1}{j}p^j(1-p)^{(n+1) - j} - (1-p)^{n+1} \bigg) \\
        \end{align*}
Hence, the Law of Total Probability implies
\begin{align*}
        \mathbb{E} \left[ \frac{1}{Y+1} \right] &= \frac{1 - (1-p)^{n+1}}{p(n+1)} \\
        & \leq \frac{1}{(n+1)p} \\
        \end{align*}
    \end{proof}
\end{lemma}

\begin{lemma}[Lemma 2 in \cite{soloff2024edge}]
\label{lem:SL-key-lemma}
Let $p_1,\dots,p_{m-1}\in [0,1]$ be fixed (non-random) and let $p_m \sim~\text{Uniform}(0,1)$. Then for the support line procedure defined by \eqref{eq:sl}, we have
\begin{align*}
    \P(p_{(R)}=p_m) \leq \frac{\tol}{m},
\end{align*}
with equality if $\tol \leq 1$.
\end{lemma}

\begin{lemma}
\label{lem:expectation-bound}
Suppose $(p_i)_{i \in \H_0}$ are independent and $p_i \sim \text{Uniform}(0,1)$ when $H_i$ is true. Then
    \begin{align*}
        \mathbb{E} \left[ \frac{\tol m_0}{m - R_1(1)}\right] \leq \frac{\tol}{1 - \tol},
    \end{align*}
    where $R_1(1)$ is the number of rejections in stage 1 of TSSL, holding $p_m$ fixed at 1.
\end{lemma}
\begin{proof}
    Consider the rejection set obtained by running TSSL on $P_{(-m)}(1)$, defined in Lemma \ref{lem:p-to-one}. Define the sets $V$ and $U$ as follows:
    \begin{align*}
        V &:= \#\{ i \leq m-1 : H_i \text{ is true}, \; p_i \leq p_{(R_1)}\} \\ 
        U &:= \#\{ i \leq m-1 : H_i \text{ is true}, \; p_i > p_{(R_1)}\}. 
    \end{align*}
We have by definition
$$(m_0 - 1) + (R_1(1) - V) \leq m$$
Because $m_0 = U + V + 1$, we can write
$$U + 1 \leq m - R_1(1)$$
By definition, $p_{(R_1)} \leq \tol$ which implies
\begin{align*}
    U \geq \# \{i \leq m-1: H_i\text{ is true},\; p_i > \tol \}
\end{align*}
Let $Y \sim \text{Binomial}(m_0 - 1, 1 - \tol)$. We can see that $Y$ must be stochastically smaller than $U$, implying that $m - R_1(1)$ is stochastically larger than $Y + 1$. Thus, by \ref{lem: expectation-bound-sub-lemma} we have $\mathbb{E} \left[ \frac{1}{Y+1} \right] \leq \frac{1}{m_0 \left( 1 - \tol \right)}$. Thus, we have 
\begin{align*}
    \mathbb{E} \left[ \frac{\tol m_0}{m - R_1(1)}\right] =  \tol m_0\mathbb{E}\left[ \frac{1}{m - R_1(1)}\right] \leq \frac{\tol}{1 - \tol}.
\end{align*}

\end{proof}



\subsection{Asymptotic formula for lfdr at the boundary}
\label{sec:appendix-asymptotic}

\begin{thm}
\label{thm:asymptotic-lfdr-control}
    Let $F^{(1)},F^{(2)},\dots$ be a sequence of continuous cdfs, where $p_i \sim F^{(i)}$, and suppose that $\|F_m - \bar{F}_m\|_\infty \to 0$ in probability as $m \to \infty$, where $F_m(t) \coloneqq \frac{1}{m}\sum_{i=1}^m 1\{p_i \leq t\}$ is the empirical cdf and $\bar{F}_m\coloneqq \frac{1}{m}\sum_{i=1}^m F^{(i)}$ is the average among the first $m$ cdfs. Further suppose that $\bar{F}_m$ is concave for each $m$ and that there exist unique values $t_1^*,t_2^* \in (0,1)$ such that
    \begin{align*}
        t_1^* &= \argmin_{t \in [0,1]} \{t - \tol \bar{F}_m(t)\}, \hspace{2em}\bar{F}_m(t^*_1) < 1,\\ 
        t_2^* &= \argmin_{t \in [0,1]} \Big\{t -  \frac{\tol\bar{F}_m(t)}{1-\bar{F}_m(t_1^*)}\Big\}.
    \end{align*}
    Let $\tau_1 \coloneqq  \argmin_{t \in [0,1]} \big\{ t - \tol F_m(t) \big\}$ and $\tau_2 \coloneqq \argmin_{t \in [0,1]} \Big\{ t - \frac{\tol F_m(t)}{1-F_m(\tau_1)} \Big\}$ be the empirical counterparts. Then
    \begin{align*}
        \lim_{m \to\infty} \left| \lfdr_m(\tau_2)- \frac{\tol \bar{\pi}_{0,m}}{1-\bar{F}_m(t_1^*)} \right| = 0,
    \end{align*}
    where $\lfdr_m(t) \coloneqq \bar{\pi}_{0,m} / \bar{f}_m(t)$, $\bar{f}_m(t) \coloneqq \frac{1}{m}\sum_{i=1}^m f^{(i)}(t)$ is the average density, and the true null proportion is denoted $\bar{\pi}_{0,m} \coloneqq \#\{i:H_i=0\}/m$.
    
\end{thm}

\begin{remark}
    Together with Lemma \ref{lem:second-stage-ineq} stated below, the above theorem shows asymptotic control of the boundary lfdr below the level $\tol/(1-\tol)$. Notably, this result does not require independence between the $p$-values; it instead requires a condition involving the convergence of the empirical cdf to the average cdf. The upshot of this result is that an explicit expression is obtained for how far below the boundary lfdr is from the upper bound of $\tol/(1-\tol)$, asymptotically.
\end{remark}

\begin{proof}
    Define
    \begin{align*}
        H_m(t) &:= \argmin_{t \in [0,1]} \Big\{ t - \frac{\tol F_m(t)}{1-F_m(\tau_1)} \Big\} \\
        H(t) &:= \argmin_{t \in [0,1]} \Big\{ t - \frac{\tol \bar{F}(t)}{1-\bar{F}(t_1^*)} \Big\}
    \end{align*}
The difference $\left\lvert H(t) - H_m(t) \right\rvert$ is equal to
\begin{align*}
    &= \tol \left\lvert \frac{\bar{F}(t)}{1-\bar{F}(t_1^*)} - \frac{F_m(t)}{1-F_m(\tau_1)} \right \rvert \\
    & \leq \tol \Big( \left\lvert \frac{\bar{F}(t)}{1-\bar{F}(t_1^*)} - \frac{F_m(t)}{1-\bar{F}(t_1^*)} \right \rvert + \left \lvert \frac{F_m(t)}{1-\bar{F}(t_1^*)} - \frac{F_m(t)}{1-\bar{F}(\tau_1)}\right \rvert + \left \lvert \frac{F_m(t)}{1-\bar{F}(\tau_1)} - \frac{F_m(t)}{1-F_m(\tau_1)}\right \rvert \Big)
\end{align*}
For the first term, we have
$$\left\lvert \frac{\bar{F}(t)}{1-\bar{F}(t_1^*)} - \frac{F_m(t)}{1-\bar{F}(t_1^*)} \right \rvert \leq \frac{1}{1-\bar{F}(t_1^*)} \|F_m - \bar{F} \|_\infty \rightarrow 0 \text{ in probability.}$$
For the second term, we first note that $\tau_1 \rightarrow t_1^*$ by Proposition 8 of \cite{soloff2024edge}. Thus, since $\bar{F}$ is continuous, we have $ \frac{1}{1-\bar{F}(\tau_1)} \rightarrow \frac{1}{1-\bar{F}(t^*_1)}$ by the Continuous Mapping Theorem. Therefore, we can write:
$$\left \lvert \frac{F_m(t)}{1-\bar{F}(t_1^*)} - \frac{F_m(t)}{1-\bar{F}(\tau_1)}\right \rvert \rightarrow 0$$
For the final term, we have
$$\left \lvert \frac{F_m(t)}{1-\bar{F}(\tau_1)} - \frac{F_m(t)}{1-F_m(\tau_1)}\right \rvert  = F_m(t)\left \lvert \frac{\bar{F}(\tau_1) - F_m(\tau_1)}{(1-\bar{F}(\tau_1))(1-F_m(\tau_1))} \right \rvert \rightarrow 0$$ because $\|F_m - \bar{F}\|_\infty \to 0$.\\
\\
We can then use the argument from Proposition 8 of \cite{soloff2024edge} to conclude $| \tau_2 - t_2^*| \to 0$. If $\bar{f}$ is continuous, we must also have $\lfdr_m(\tau_2) \to \lfdr_m(t_2^*)$.\\
\\
\end{proof}

\begin{lemma}
\label{lem:second-stage-ineq}
Suppose that the average non-null density $\frac{1}{m_1}\sum_{i: H_i\text{ is false}} f^{(i)}(t)$ is positive and non-increasing on $[0,1]$, and the solution $t_1^*$ to $\bar{\pi}_{0,m} / \bar{f}_m(t^*_1) = \bar{\pi}_0\tol$ is unique, where this notation was defined in Theorem \ref{thm:asymptotic-lfdr-control}. Then we have
\begin{align*}
    \frac{\tol \bar{\pi}_{0,m}}{1-\bar{F}_m(t_1^*)} \leq \frac{\tol}{1-\tol}.
\end{align*}
\end{lemma}
\begin{proof}
In this proof, we slightly abuse notation and let $\bar{f}_1$ denote the average non-null density (while $\bar{f}_m$ still denotes the overall average density) and $\bar{F}_1$ denote the average non-null cdf. Since $\bar{f}_1$ is positive and non-increasing, we can write:
\[
\bar{F}_1(t^*_1) = \int_0^{t^*_1} \bar{f}_1(t)\,dt \ge \int_0^{t^*_1} \bar{f}_1(t^*_1)\,dt = t^*_1 \cdot \bar{f}_1(t^*_1)
\]
Thus, we have
\begin{align}
\label{ineq:tstar1}
t^*_1 \le \frac{\bar{F}_1(t^*_1)}{\bar{f}_1(t^*_1)} 
\end{align}
Because $t^*_1$ satisfies $\bar{\pi}_{0,m} / \bar{f}(t^*_1) = \bar{\pi}_{0,m}\tol$, 
we can also write:
\[
\bar{\pi}_{0,m} + (1 - \bar{\pi}_{0,m})\bar{f}_1(t^*_1) = \bar{f}(t^*_1) = \tol^{-1},
\]
Rearranging, we get:
\begin{align}
\label{eq:tstar2}
\bar{f}_1(t^*_1) = \frac{1 - \bar{\pi}_{0,m} \tol}{(1 - \bar{\pi}_{0,m}) \tol} .
\end{align}
We now have:
\begin{align*}
    \bar{F}(t^*_1) &= \bar{\pi}_{0,m} t_1^* + (1 - \bar{\pi}_{0,m})\bar{F}_1(t_1^*) \\
    & \leq \bar{F}_1(t_1^*)\left(\frac{\bar{\pi}_{0,m}}{\bar{f}_1(t_1^*)} + \left(1 - \bar{\pi}_{0,m} \right)\right) \space \tag{by \eqref{ineq:tstar1}} \\ 
    &= \bar{F}_1(t_1^*)\left(1 - \bar{\pi}_{0,m} \right) \left(\frac{\bar{\pi}_{0,m} \tol}{1 - \bar{\pi}_{0,m} \tol} + 1 \right) \space \tag{by \eqref{eq:tstar2}} \\
    & \leq \left(1 - \bar{\pi}_{0,m} \right) \left(\frac{\bar{\pi}_{0,m} \tol}{1 - \bar{\pi}_{0,m}} + 1 \right) \\
    &= 1 - \bar{\pi}_{0,m} \left(1 - \tol \right)
\end{align*}
We can rearrange this to obtain:
\[
\frac{\bar{\pi}_{0,m}}{1 - \bar{F}_m(t_1^*)} \le \frac{1}{1 - \tol} ,\quad 
\]
as desired.
\end{proof}



\bibliographystyle{dcu}
\bibliography{reference}

\end{document}